\documentclass[sigconf]{acmart}
\pdfoutput=1
\usepackage{amsmath}
\usepackage{amssymb}
\usepackage{bm}
\usepackage{nicefrac}
\usepackage{booktabs}
\usepackage{array}
\usepackage{multirow}
\usepackage{threeparttable}
\usepackage{makecell}
\usepackage[procnumbered,ruled,vlined,linesnumbered]{algorithm2e}
\usepackage{siunitx}
\usepackage{stfloats}
\usepackage{graphicx}
\usepackage{subfigure}
\usepackage{hyperref}
\usepackage{textcomp}
\usepackage{array}
\usepackage{epstopdf}
\usepackage{framed}
\usepackage[most]{tcolorbox}
\usepackage[normalem]{ulem}
\usepackage{xcolor}

\def\defeq{\stackrel{\mathrm{def}}{=}}

\def\Gcal{\mathcal{G}}
\def\Ical{\mathcal{I}}

\def\sizeof#1{\left|#1  \right|}
\def\abs#1{\left|#1  \right|}
\newcommand\one{\boldsymbol{1}}
\newcommand{\rea}{\mathbb{R}}
\newcommand\XXtil{\boldsymbol{\mathit{\tilde{X}}}}
\newcommand\YYtil{\boldsymbol{\mathit{\tilde{Y}}}}

\renewcommand\vv{\boldsymbol{\mathit{v}}}
\newcommand\yy{\boldsymbol{\mathit{y}}}
\newcommand\zz{\boldsymbol{\mathit{z}}}

\newcommand\bb{\boldsymbol{\mathit{b}}}

\newcommand\ee{\boldsymbol{\mathit{e}}}
\newcommand\qq{\boldsymbol{\mathit{q}}}

\newcommand\vs{\boldsymbol{\mathit{s}}}
\newcommand\hh{\boldsymbol{\mathit{h}}}

\newcommand\LL{\boldsymbol{\mathit{L}}}
\newcommand\AAA{\boldsymbol{\mathit{A}}}
\newcommand\BB{\boldsymbol{\mathit{B}}}

\newcommand\DD{\boldsymbol{\mathit{D}}}
\newcommand\EE{\boldsymbol{\mathit{E}}}
\newcommand\PP{\boldsymbol{\mathit{P}}}
\newcommand\MM{\boldsymbol{\mathit{M}}}
\newcommand\QQ{\boldsymbol{\mathit{Q}}}

\newcommand\II{\boldsymbol{\mathit{I}}}

\newcommand\XX{\boldsymbol{\mathit{X}}}

\newcommand\OM{\boldsymbol{\Omega}}
\newcommand\RR{\boldsymbol{\mathit{R}}}
\newcommand\SSS{\boldsymbol{\mathit{S}}}

\newcommand{\eps}{\epsilon}

\DeclareMathOperator*{\argmin}{arg\,min}

\newtheorem{definition}{Definition}[section]
\newtheorem{lemma}{lemma}[section]
\newtheorem{theorem}{theorem}[section]
\newtheorem{proposition}{Proposition}

\newcommand{\PreserveBackslash}[1]{\let\temp=\\#1\let\\=\temp}
\newcolumntype{C}[1]{>{\PreserveBackslash\centering}p{#1}}
\newcolumntype{R}[1]{>{\PreserveBackslash\raggedleft}p{#1}}
\newcolumntype{L}[1]{>{\PreserveBackslash\raggedright}p{#1}}

\newenvironment{fminipage}%
{\begin{Sbox}\begin{minipage}}%
		{\end{minipage}\end{Sbox}\fbox{\TheSbox}}

\def\kh#1{\left( #1 \right)}

\def\ceil#1{\left\lceil #1 \right\rceil}
\newcommand\Otil{\widetilde{O}}

\def\norm#1{\| #1 \|}
\newcommand{\naiveGreedy}{\textsc{Greedy}}
\newcommand{\FastGreedy}{\textsc{Greedy-AC}}
\newcommand{\FastGreedyy}{\textsc{AC}}
\newcommand{\Approx}{\textsc{Estimat}}
\newcommand{\Solver}{\textsc{Estimator}}
\newtheorem{problem}{Problem}

\newfont{\nset}{msbm10}
\newcommand{\removelatexerror}{\let\@latex@error\@gobble}

\DontPrintSemicolon
\SetKw{KwAnd}{and}
\SetProcNameSty{textsc}
\SetFuncSty{textsc}
\SetKwInOut{Input}{Input\ \ \ \ }
\SetKwInOut{Output}{Output}
\usepackage{tabularx}
\usepackage{stfloats}

\AtBeginDocument{%
	\providecommand\BibTeX{{%
			\normalfont B\kern-0.5em{\scshape i\kern-0.25em b}\kern-0.8em\TeX}}}

\copyrightyear{2022}
\acmYear{2022}
\setcopyright{acmcopyright}
\acmConference[KDD '22] {Proceedings of the 28th ACM SIGKDD Conference on Knowledge Discovery and Data Mining}{August 14--18, 2022}{Washington, DC, USA.}
\acmBooktitle{Proceedings of the 28th ACM SIGKDD Conference on Knowledge Discovery and Data Mining (KDD '22), August 14--18, 2022, Washington, DC, USA}
\acmPrice{15.00}
\acmISBN{978-1-4503-9385-0/22/08}
\acmDOI{10.1145/3534678.3539469}
\begin{document}
	%\fancyhead{}
	%%
	%% The "title" command has an optional parameter,
	%% allowing the author to define a "short title" to be used in page headers.
	\title{A Nearly-Linear Time Algorithm for Minimizing Risk of Conflict in Social Networks}
	
	 	\author{Liwang Zhu}
	 	\affiliation{%
	 		\institution{Fudan University}
	 		%\streetaddress{1 Th{\o}rv{\"a}ld Circle}
	 		\city{Shanghai}
	 		\country{China}}
	 	\email{19210240147@fudan.edu.cn}

	 	\author{Zhongzhi Zhang}
	 	\affiliation{%
	 		\institution{Fudan University}
	 		%\streetaddress{1 Th{\o}rv{\"a}ld Circle}
	 		\city{Shanghai}
	 		\country{China}}
	 	\email{zhangzz@fudan.edu.cn}
	
%	\author{Liwang Zhu and Zhongzhi Zhang}
	\authornote{Zhongzhi Zhang is the corresponding author. Both authors are with Shanghai Key Laboratory of Intelligent Information Processing, School of Computer Science, Fudan University, Shanghai 200433.}
%	\affiliation{%
%		\institution{Shanghai Key Laboratory of Intelligent Information Processing, Fudan University, Shanghai 200433, China}
%		\institution{School of Computer Science, Fudan University, Shanghai 200433, China}
%		\city{}
%		\country{}
%	}
%	\email{{19210240147,zhangzz}@fudan.edu.cn}
	%\streetaddress{1 Th{\o}rv{\"a}ld Circle}
	% 		\renewcommand{\thefootnote}{*}
	% \footnotetext[1]{Zhongzhi Zhang is the corresponding author, who is also with Shanghai Blockchain Engineering Research Center, as well as Research Institute of Intelligent Complex Systems, Fudan University, Shanghai 200433.} %
	\begin{abstract}
		%The resistance and controversy of society over controversial social issues has been the subject of study in social sciences for decades.  Individuals interact with their acquaintances and exchange opinions whereas the opinions of nodes often do not reach consensus, leading to resistance, controversy and other important phenomena, which have been the subject of many recent works. People’s opinions are fundamentally influenced by their interactions with other individuals. In the computer science and social science literature, the interplay between network structure and opinion formation has been extensively studied to understand phenomena of consensus and polarization, as well as the role of interventions to unidirectionally shift the overall sentiment on a particular issue [1, 3, 4, 8, 10, 11].
		% (e.g., through education, exposure to diverse viewpoints, or incentives) 
		Concomitant with the tremendous prevalence of online social media platforms, the interactions among
		individuals are unprecedentedly enhanced. People are free to interact with acquaintances, express and exchange their own opinions through commenting, liking, retweeting on online social media, leading to resistance, controversy and other important phenomena over controversial social issues, which have been the subject of many recent works. In this paper, we  study the problem of minimizing risk of conflict in social networks by modifying the initial opinions of a small number of nodes. We show that the objective function of the combinatorial optimization problem is monotone and supermodular. We then propose a na\"{\i}ve greedy algorithm with a $(1-1/e)$ approximation ratio that  solves the problem in cubic  time. To overcome the computation challenge for large networks, we further integrate several effective approximation strategies to provide a nearly linear time algorithm with  a $(1-1/e-\epsilon)$ approximation ratio for any error parameter $\epsilon>0$. Extensive experiments on various real-world  datasets  demonstrate both the efficiency and effectiveness of our algorithms. In particular, the fast one scales to large networks with more than two million nodes, and achieves up to $20\times$ speed-up over the state-of-the-art algorithm.
		%Here, rather than  manipulating all individuals’ innate opinions, we study the optimization problem of influencing the opinion of a small number of individuals.
	\end{abstract}
\begin{CCSXML}
	<ccs2012>
	<concept>
	<concept_id>10003752.10003809.10003635</concept_id>
	<concept_desc>Theory of computation~Graph algorithms analysis</concept_desc>
	<concept_significance>500</concept_significance>
	</concept>
	<concept>
	<concept_id>10003752.10010070.10010099.10003292</concept_id>
	<concept_desc>Theory of computation~Social networks</concept_desc>
	<concept_significance>500</concept_significance>
	</concept>
	<concept>
	<concept_id>10003752.10003809.10003716.10011136</concept_id>
	<concept_desc>Theory of computation~Discrete optimization</concept_desc>
	<concept_significance>500</concept_significance>
	</concept>
	<concept>
	<concept_id>10002951.10003227.10003351</concept_id>
	<concept_desc>Information systems~Data mining</concept_desc>
	<concept_significance>500</concept_significance>
	</concept>
	</ccs2012>
\end{CCSXML}

\ccsdesc[500]{Theory of computation~Graph algorithms analysis}
\ccsdesc[500]{Theory of computation~Social networks}
\ccsdesc[500]{Theory of computation~Discrete optimization}
\ccsdesc[500]{Information systems~Data mining}

\keywords{Opinion dynamics, graph algorithm, resistance, controversy, social network,  discrete optimization}

\maketitle

	\section{Introduction}
	
	It has been extensively studied in the social science literature how opinions evolve and shape through social interactions between individuals with potentially differing opinions~\cite{DeMo74,FrJo90}. In the current digital age, the tremendous prevalence of online social networks and social media provide unprecedented access to social interactions, expression and exchange of opinions, leading to fundamental changes of ways people share and formulate opinions. The uninhibited access to information and expression of opinions leads to the emergence or reinforcement of various social phenomena, such as polarization, disagreement, controversy, and resistance, which are signified in~\cite{ChLiDe18} by the term conflict in a more generic manner. For example, users in the virtual world tend to create connections with like-minded individuals, which separates individuals into groups forming “echo-chambers” or “filter bubbles”.  Individuals in different groups have little even no communication with each other, whose opinions do not reach consensus but are opposing, leading to and reinforcing  polarization and disagreement.

	The identification~\cite{XuBaZh21}, quantification~\cite{ChLiDe18,MuMuTs18}, and optimization~\cite{MaTeTs17,MuMuTs18,HaMeCrRi21,GaDeGiMa17,BiKlOr15,GaKlTa20} of conflict are fundamental tasks behind a myriad of high-impact data mining applications, and thus have received considerable attention. Since conflict has a corrosive and detrimental risk to the functioning of communities and societies~\cite{MaTeTs17}, it is thus of significance to reduce the risk of conflict through some targeted interventions, minimizing or mitigating those negative effects. In this paper, we focus on optimizing two primary measures of conflict, controversy and resistance, with the former also called polarization in~\cite{MaTeTs17,MuMuTs18}. Specifically, we minimize controversy and resistance by changing the opinions of a small number of individuals, which can be achieved by raising awareness and enforcing education of individuals, among other typical strategies or means~\cite{MaTeTs17,MuMuTs18,GaKlTa20}.

	% When people only get information that corroborates their own opinions and communicate only with similar mindsets people, Disagreement~\cite{MuMuTs18,DaGoLe13}  characterizes how much acquaintances disagree in their opinions, globally across the network. Polarization~\cite{MuMuTs18,MaTeTs17,DaGoLe13} measures how equilibrium expressed opinions deviate from the average. However, existing recommender system, trained on real-data, with the goal to increase user engagement may stop the user from being exposed to diverse opinions and naturally end up creating “echo-chambers”. In other words, the recommended links minimize disagreement may lead to greater polarization~\cite{MuMuTs18} since connections between users with similar mindsets are preferred for such system. Yet, exposure to diverse content is necessary to obtain a complete picture about a topic~\cite{HaMeCrRi21}. Thus, there is a need for a radically different approach to suggest links that decrease both disagreement and polarization, which  motivates our work.
	%We are interested, in particular, in minimizing the sum of disagreement and polarization through targeted link recommendation.
	%Thus, the trade-off between disagreement and polarization motivates our work.
	
	\textbf{Shortcomings in the state-of-the-art.} The study of optimizing conflict in social media by convincing a small number of people to adopt a different stand is not new. In~\cite{MaTeTs17},  an algorithm called \emph{BOMP} was proposed to reduce controversy by selecting a group of $k$ individuals in a social network with $n$ nodes and $m$ edges, and convincing them to change their initial opinions to $0$. The computational complexity of \emph{BOMP} is $O(kn^2)$. As an input of algorithm \emph{BOMP}, the forest matrix~\cite{GoDrRo81,Ch08} is assumed to have been pre-computed in~\cite{MaTeTs17}. Actually, the computation of forest matrix involves matrix inverse, which is time-consuming and requires time $O(n^3)$. To tackle this computation challenge, we develop a nearly linear time algorithm with respect to $m$, the number of edges. In addition, \emph{BOMP} is designed for minimizing controversy, while our approach is also applicable to the optimization of resistance.

	%\textbf{Contributions.} In this paper, we address the following optimization problem: given a social network with $n$ nodes and $m$ edges, and a budget value $k$, how to strategically identify $k$ individuals and convince them to change their initial opinions (in our model,  their opinions value are set to zero) will minimize the conflict measures including controversy and resistance. Although these two optimization problems are different, we show that both of the objective functions  are supermodular and monotone. Exploiting the diminishing returns property of the problem, we  propose two greedy algorithm to solve the problem. The former has a $(1-1/e)$ approximation ratio that solves the problem in cubic  time, while the latter  provides a $(1-1/e-\eps)$ approximation ratio with computation complexity $\Otil (mk\eps^{-2})$ for any $\eps>0$,  where  $ \eps>0$ is the error parameter and the $\Otil (\cdot)$ notation suppresses the ${\rm poly} (\log n)$ factors. We evaluate our theoretical and algorithmic performance by executing extensive experiments on various real-world networks, which show that our algorithms are efficient and effective. In particular, we empirically demonstrate the superiority of our fast algorithm over BOMP (the state-of-the-art algorithm for reducing the controversy) in terms of the computational efficiency based on $18$ real-world graphs. To be specific,  the speed-up by our algorithm is up to $20\times$, while its estimation accuracy is comparable to that of BOMP. 
	
	\textbf{Contributions.} In this paper, we address the following optimization problem: given a social network with $n$ nodes and $m$ edges, a vector $\vs$ of initial opinions, and a budget value $k$, how to strategically identify $k$ nodes and change their initial opinions to zero, in order to minimize two conflict measures, controversy and resistance. Our main contributions include the following three aspects. First, we unify the two optimization objectives into one framework, and show that the unified objective function is supermodular and monotone. Then, based on the obtained properties of the objective function, we propose two greedy algorithms, $\naiveGreedy$ and $\FastGreedy$, to solve the problem. $\naiveGreedy$ has a $(1-1/e)$ approximation ratio with computation complexity of $O(n^3)$, while $\FastGreedy$ has a $(1-1/e-\eps)$ approximation ratio with computation complexity $\Otil (mk\eps^{-2})$ for any $\eps>0$,  where  $ \eps>0$ is the error parameter and the $\Otil (\cdot)$ notation suppresses the ${\rm poly} (\log n)$ factors. Finally, we evaluate the performance of our algorithms by executing extensive experiments on various real-world networks, which show that $\FastGreedy$ is as effective as $\naiveGreedy$ and \emph{BOMP}, all of which outperform several baseline strategies. Moreover, $\FastGreedy$ is more efficient than $\naiveGreedy$ and \emph{BOMP}, with   $\FastGreedy$ achieving up to $20\times$ speed-up over $\naiveGreedy$ and \emph{BOMP} on moderately sized networks with $24$ thousand nodes. In particular, $\FastGreedy$ is scalable to large networks with more than two million nodes.

	\section{Related work}
	%We review the related literature from the following three perspectives, including (a) modeling opinion dynamics, (b) optimization problems in opinion dynamics and (c) link recommendation strategy.
	In this section, we  briefly review the related literature.

	\textbf{Optimization of polarization and controversy.} Due to the negative effects of polarization and controversy, a lot of works have been devoted to designing strategies to decrease these two correlated quantities. For instance, in~\cite{HaMeCrRi21} and~\cite{GaDeGiMa17}, link addition was considered to maximally reduce the polarized bubble radius and controversy, respectively. Both of these two existing works focus on graph-theoretic measures of polarization or controversy, which do not take the initial opinions of individuals into consideration. Furthermore, both of them exploit the strategy of the link addition to achieve the goal, instead of the modification of initial opinions.
	
	%Given the negative effects of polarization on the societies, there has been much work that focuses on methods to decrease them. Link insertion is considered to maximally reduce the polarized bubble radius~\cite{HaMeCrRi21} and  observed polarization~\cite{GaDeGiMa17}, respectively. However, both of them focus on graph-theoretic measures of polarization rather than taking the opinions of individuals into consideration. Furthermore, they consider the addition of links, rather than the moderation of opinions.
	
	The closest to our work lies that of~\cite{MuMuTs18} and~\cite{MaTeTs17}. In~\cite{MuMuTs18}, the $\ell_2$ norm of $\overline{\zz}=\zz - \frac{\zz^\top  \textbf{1}}{n} \textbf{1}$ was used to represent polarization, where $\zz$ is the vector of equilibrium expressed opinion, and $\textbf{1}$ is the all-ones vector. It is close to our considered controversy, except that we do use the mean-centered opinion vector $\overline{\zz}$. Moreover, the method of modifying initial opinions was applied to minimize the sum of polarization and disagreement in~\cite{MuMuTs18}, where all nodes’ initial opinions can be manipulated. In contrast, we only modify a fixed number of nodes’ opinions to achieve our goals, which is more realistic. In the context of algorithms, to reduce polarization by changing the opinions of a small number of nodes, algorithm  \emph{BOMP} was proposed in~\cite{MaTeTs17}  with an actual  complexity of $O(n^3+kn^2)$, which is in sharp contrast to that of our nearly-linear time algorithm $\FastGreedy$.  Last but not the least, in addition to polarization or controversy, $\FastGreedy$ is also applicable to the optimization of resistance.

	%The  closest to our work lies that of~\cite{MuMuTs18} and~\cite{MaTeTs17}. The former work~\cite{MuMuTs18} focuses on the $\ell_2$ norm of the mean-centered equilibrium vector $\overline{\zz}=\zz - \frac{\zz^\top  \textbf{1}}{n} \textbf{1}$ under the Friedkin-Johnsen model, which is close to our measure of controversy, except that we do not mean center the opinion vector. The above work propose controlling internal opinions to minimize the sum of polarization and disagreement. However, they assume that all individuals’ innate opinions can be manipulated, which is hard to achieve in reality. Here, we only control a few individuals’ internal opinions. In addition, our nearly-linear-time algorithms run faster than existing polynomial-time ones. An algorithm called \emph{BOMP} is proposed~\cite{MaTeTs17}  to reduce polarization by influencing the opinion of a small number of people, resulting in  $O(n^3+kn^2)$ complexity in total. However, their algorithm relies primarily on the connection with sparse approximation problem while we use  spectral techniques. Moreover, compared with the \emph{BOMP}, our linear-nearly time algorithm runs faster and is able to apply to optimize the resistance measurement. In the experiment section, we compare our approach to that of~\cite{MaTeTs17}.
	
	\textbf{Other optimization problems in opinion dynamics.} Other optimization problems related to opinion dynamics have also been formulated and studied for different objectives. For example, a long line of work has been devoted to the problem of influence or opinion maximization by using different strategies, including identifying a fixed number of individuals and changing their expressed opinions to 1~\cite{GiTeTs13}, changing the  initial opinions of agents~\cite{XuHuWu20},  modifying susceptibility to persuasion~\cite{AbKlPaTs18,ChLiSo19}, and so on. Furthermore,~\cite{TuAsCiGi20} considered the problem of allocating seed users to two opposing campaigns with an aim to maximize the expected number of users who are co-exposed to both campaigns. 
	
	Another major and increasingly important focus of research is optimizing other social phenomena or related quantities, such as maximizing the diversity~\cite{MaPa19,MaTuGi20}  and minimizing disagreement~\cite{GaKlTa20,YiSt20}. In~\cite{BiKlOr15}, the operation of edge addition was exploited  in order to reduce the social cost, which is the weighted sum of internal and external conflicts. The strategy of adding a limited number of edges was also applied in~\cite{AmSi19} to fight opinion control in social networks.
	
	%\textbf{Other optimization problems in opinion dynamics.} Recently, several optimization problems related to opinion dynamics have been formulated and studied for different objectives. For example,  a long line of work  has been devoted to influence and opinion maximization by using different strategies, such as identifying a fixed number of individuals and setting their expressed opinions to 1~\cite{GiTeTs13}, changing agent's internal opinions~\cite{XuHuWu20}, as well as modifying individuals'  susceptibility to persuasion~\cite{AbKlPaTs18,ChLiSo19}.~\cite{TuAsCiGi20} studies the problem of allocating seed users to opposing campaigns with the goal to maximize the expected number of users who are co-exposed to both campaigns. These studies have far-reaching implications in product marketing, public health campaigns, and political candidates. 
	
	%Another major and increasingly important focus of research is optimizing some social phenomena, such as maximizing the diversity~\cite{MaPa19,MaTuGi20},  and disagreement~\cite{GaKlTa20,YiSt20}.  In~\cite{BiKlOr15}, addition of edges is discussed  in order to reduce  the social cost, namely the lack of agreement. In~\cite{AmSi19}, the authors strategically fights opinion control in social networks  by recommending a limited number of edges to individuals. 
	
	\textbf{Opinion mining.} In this work, the initial opinions of nodes are given as input, which are used to minimize controversy and resistance. By applying the techniques of opinion mining and sentiment analysis~\cite{ZhLi17}, the expressed opinion of a node is readily observable in a social network. However, the initial opinions of nodes are not accessible, which are often hidden. Although~\cite{DaGoPaSa13} proposed a nearly-optimal sampling algorithm for estimating the average of initial opinions in social networks, which cannot be used to evaluate the initial opinion of an individual. Including an opinion mining algorithm as the first step of the pipeline could extend our work for optimizing controversy and resistance.

	\section{Preliminaries}\label{S2}
	%In this section, we give a brief introduction to some essential concepts and tools, in order to facilitate the description of the problem and its related greedy algorithms.
	In this section, we present definitions and relevant results to facilitate  the description of our problem and  development of our greedy algorithms.
	%\subsection{Notations}
	%We use normal lowercase letters like $a, b, c$ to represent scalars in $\rea$,  normal uppercase letters like $A, B, C$ to represent sets,  bold lowercase letters like $\aaa, \bb, \cc$ to represent vectors,  and bold uppercase letters like $\AAA, \BB, \CC$ to represent matrices.
	%Let $\ee_i$  denote the $i^{\rm th}$ standard basis vector of appropriate dimension.  Let $\one$ and $\JJ$ denote,  respectively,  the vector and the matrix of appropriate dimensions with all entries being ones.  Then,   $\JJ=\one\one^\top$.  Let $\mathbf{0}$ and $\mathbf{O}$ denote,  respectively,  the vector and matrix of appropriate dimensions with all entries being zeros.  Let $\aaa^\top$ and $\AAA^\top$  denote,  respectively,  transpose of  vector $\aaa$ and matrix  $\AAA$. For simplicity, we use  $a_{ij}$ to denote the entry of $\AAA$ at $i^{\rm th}$ row and $j^{\rm th}$ column. We use  $\AAA_i$ to denote the $i^{\rm th}$ row of $\AAA$. We use $\aaa_i$ to denote  the $i^{\rm th}$ element of vector $\aaa$. For two matrices $\AAA$ and $\BB$,  we write $\AAA \preceq \BB$ to denote
	%that $\BB - \AAA$ is positive semidefinite. 
	\subsection{Graph and Related Matrices}
	Consider a connected, undirected, simple graph (network) $\Gcal= (V,E)$ with $n$ nodes and $m$ edges,  where $V=\{v_1,v_2,\cdots,v_n\}$ is the set of vertices/nodes,  $E\subseteq V \times V=\{e_1,e_2,\cdots,e_m\}$ is the set of edges. In the sequel, we will use $v_i$ and $i$  interchangeably to represent node $v_i$ if incurring no confusion.
	
	The adjacency relation of all nodes in $\Gcal$ is characterized by  its adjacency matrix $\AAA=(a_{ij})_{n \times n}$. If nodes $i$ and $j$ are adjacent by an edge $e$, then $a_{ij}= a_{ji}=1$; $a_{ij}=a_{ji}=0$ otherwise. Let $N_i$ be the set of neighbours of node $i$ satisfying $N_i=\{j| \{i, j\}\in E\}$. Then, the degree $d_i$ of a node $i$ is $d_i=\sum_{j=1}^n a_{ij}=\sum_{j\in N_i} a_{ij}$, and the diagonal degree matrix of  $\Gcal$ is defined as ${\DD} = {\rm diag}(d_1, d_2, \ldots, d_n)$.
	
	The Laplacian matrix of $\Gcal$ is defined to be ${\LL}={\DD}-{\AAA}$. There is also an alternative construction of $\LL$ by using the incidence matrix $\BB \in \mathbb{R}^{|E| \times |V|}$, an $m\times n$ signed edge-node incidence matrix. For each edge $e\in E$ and node $ v\in V$, the element $b_{ev}$ of $\BB$  is defined as follows: $b_{e v}=1$ if  $v$ is the head of $e$, $b_{ev}=-1$ if  $v$ is the tail of $e$, and $b_{ev}=0$ otherwise. For an edge $e\in E$ with two end nodes $i$ and $j$, the row vector of $\BB$ corresponding to  $e$ can be written as $\bb_{ij}\triangleq \bb_{e}=\ee_i-\ee_j$ where $\ee_i$  denotes the $i$-th standard basis vector of appropriate dimension. Then the Laplacian matrix $\LL$ of $\Gcal$ can also be represented as $\LL = \BB^\top  \BB$,  indicating that $\LL$ is  symmetric and positive semidefinite. 
	
	The Laplacian matrix $\LL$ of a connected graph $\Gcal$ has a unique zero eigenvalue. Let $0<\lambda_1\le\lambda_2\le\cdots\le\lambda_{n-1}$ be the nonzero eigenvalues of  $\LL$ of a connected graph $\Gcal$. Let $\lambda_{\max}$ and  $\lambda_{\min}$  be, respectively, the maximum and nonzero minimum  eigenvalue of  $\LL$. Then,  $\lambda_{\max}= \lambda_{n-1}\leq n\, $~\cite{SpSr11}, and $\lambda_{\min}=\lambda_{1}\geq 1/ n^2 $~\cite{LiSc18}. 
	
	The forest matrix of graph $\Gcal$ is defined as $\OM=(\II+\LL)^{-1}=(\omega_{ij})_{n\times n}$~\cite{GoDrRo81,Ch08}. For  an arbitrary pair of nodes $i$ and $j$ in graph  $\Gcal$, $\omega_{ij}\geq 0$ with equality  if and only if there is no path between  $i$ and $j$~\cite{Me97}. Matrix $\OM$ is a doubly stochastic~\cite{ChSh97,ChSh98},  satisfying $\OM\one=\one$ and  $\one^{\top}\OM=\one^{\top}$ where $\one$ denotes the all-ones vector.
	\subsection{Greedy Algorithm For Set Function}
	We first give the definitions of monotone and supermodular set functions. For a set $T$ and an element $u\notin T$, we use $T+u$ to denote the set $T \cup \{u\}$. For a finite set $X$,  we use $2^X$ to denote the set of all subsets of $X$.
	%[Monotonicity and Supermodularity]
	\begin{definition}
		A set function $f:2^X\rightarrow \rea$ is monotone nonincreasing if $f(T) \ge f(W)$ holds for all $T \subseteq W \subseteq X$, and $f$ is supermodular if $f(T) -f(T+u) \ge f(W) - f(W+u)$ holds for all $T \subseteq W \subseteq X$ and $u\in X\backslash W$.
	\end{definition}
	%It has been shown in~\cite{Al03} that there is another equivalent definition of submodularity. We give this definition in the following fact:
	%\begin{fact}
	%\label{fac:sub}
	%A set function $f:2^V\rightarrow \rea$ is submodular if and only if for all $T, W \subseteq 2^V$, $f(T) + f(W) \ge f(T\cup W) + f(T\intersect W)$. 
	%\end{fact}
	Many network topology design problems can be formulated as minimizing a monotone set function over a $k$-cardinality constraint. Formally the problem can be described as follows: find a subset $T^\ast$ satisfying $T^\ast\in\argmin_{|T|=k}f(T)$, where $f$ is a non-increasing supermodular set function. 
	
	Exhaustive search takes exponential time to obtain the optimal solution to these combinatorial optimization problems, which makes it intractable even for moderately sized networks. However, utilizing the diminishing returns property, a na\"{\i}ve greedy algorithm~\cite{NeWoFi78} has become a prevalent choice for solving such optimization  problems with a theoretical performance guarantee. %, which is stated by  the following theorem~. 
	%, and nevertheless performs well empirically even when the object function deviates from being submodular
	
	\begin{theorem}\cite{NeWoFi78}
		\label{th:subg}
		Let $T_{\rm opt}$ be the optimal solution to the above problem and  $f(T_g)$ the output of the na\"{\i}ve greedy algorithm corresponding to the subset $T_g$. If $f$ is supermodular and non-increasing, then the greedy algorithm guarantees  a near-optimal solution as:
		%\begin{equation}
		$f(\emptyset) -f(T_g)   \ge (1-1/e)(f(\emptyset) -f(T_{\rm opt}))$.
		%\end{equation}
	\end{theorem}

	\section{Model and Related Measures}\label{S3}
	
	In this section, we briefly introduce the Friedkin-Johnsen (FJ) model for opinion formation, as well as the definitions and measures for resistance and controversy. 
	%In this section, we briefly discuss the FJ model of opinion formation, as well as the definitions and measures for disagreement and polarization.
	\iffalse
	This section is denoted to brief introduction to the  FJ  model of  opinion formation, as well as the definitions and measures for  resistance, disagreement, polarization, and controversy, relying on this popular model. Particularly, we give an explanation and some properties of equilibrium  expressed opinions of the FJ model, using the  forest matrix.
	\fi
	\subsection{Opinion Formation  Model}\label{FJ}
	%FJ model and DeGroot model are the most two popular opinion dynamics models, in fact, FJ model~\cite{FrJo90} is an extension of the DeGroot’s~\cite{DeMo74}. In the DeGroot model, every individual holds only one opinion, updated as the weighted average of its neighbors. Friedkin and Johnsen extend the model by mixing each individual’s internal opinion into the averaging dynamics. 
	
	Opinion formation model social learning processes in various disciplines~\cite{DoZhKoDiLi18,AnYe19,SeGrSqRa19,JiMiFrBu15}. In the past decades, numerous relevant models for opinion dynamics have been proposed~\cite{AbKlPaTs18,GiTeTs13,BiKlOr15,RaFrTeIs15,DaGoPaSa13}. Here, we adopt the popular FJ model~\cite{FrJo90}, where each node $i \in V$ has two opinions: internal (or innate) opinion and  expressed opinion, both in the interval $[0,1]$. For each node $i$, its internal opinion denoted by $\vs_i$ remains unchanged. Let $\zz_i(t)$  be the  expressed opinion of node $i$ at time $t$. Its updating rule is defined as
	\begin{equation}\label{FJmodel}
	\zz_{i}(t+1)=\frac{\vs_{i}+\sum_{j \in N_i} a_{i j} \zz_{j}(t)}{1+\sum_{j \in N_i} a_{i j}}.
	\end{equation}
	Let $\vs=(\vs_1,\vs_2,\ldots,\vs_n)^\top$ and $\zz=(\zz_1, \zz_2,\ldots,\zz_n)^\top$ be the initial  opinion vector and equilibrium  expressed opinion vector, respectively. It has been shown in~\cite{BiKlOr15} that
	\begin{equation}\label{FJmodel02}
	\zz=(\II+\LL)^{-1}\vs\,,
	\end{equation}
	which indicates  that the equilibrium  expressed opinion of every node is  determined by the forest matrix $\OM=\left(\II+\LL\right)^{-1}$ and initial  opinion vector $\vs$. For for each $i \in V$,  $\zz_i=\sum^n_{j=1}  \omega_{ij}\vs_j $, which is a weighted average of initial  opinions of all nodes, with the weight for  opinion $\vs_j$ being  $\omega_{ij}$.  Since  $\OM$  is  doubly stochastic and $\vs_i \in [0,1]$ for all $i=1,2\ldots,n$, it follows that $\zz_i  \in [0,1]$ for every node $i\in V$. 
	%In addition, $\sum^n_{i=1}  \zz_i=\sum^n_{i=1}\vs_i $ i.e., the  overall expressed opinion is equal to the overall internal opinion, in spite that the equilibrium  expressed opinion for a single node may be different from its  internal  opinion. This conservation law is independent  of the network structure. 
	%In this sense, we provide an interpretation and some properties  of  equilibrium  expressed opinion vector $\zz$ according to the  forest matrix.
	
	\subsection{Measures of Conflict}\label{Sec1}
	In the FJ model,  the equilibrium expressed opinions  often do not reach consensus, leading to controversy, resistance  and other important phenomena.  As in~\cite{ChLiDe18}, in this paper we use term conflict in a more generic manner to signify controversy or resistance. We next survey the measures of controversy and resistance, and  discuss how they can be computed using matrix-vector operations. 
	%individuals interact with their acquaintances and exchange opinions, while
	%Arguably the most relevant measure in practice, the external resistance measure quantifies the extent to which the expressed opinions of neighbors are in disagreement with each other. Formally:
	%\begin{definition}[Conflict~\cite{MuMuTs18,DaGoLe13}]
	%	For a graph $\Gcal= (V,E)$ with expressed opinion vector $\zz$, its  disagreement $D(\Gcal,\zz)$  is defined as
	%	\begin{equation}\label{eq:dfn_disagree}
	%	D(\Gcal,\zz) =   \sum\limits_{(i,j) \in E} (\zz_i-\zz_j)^2.
	%	\end{equation}
	%\end{definition}
	%Disagreement $D(\Gcal,\zz)$ characterizes how much acquaintances disagree in their opinions, globally across the network.
	
	Controversy  quantifies how much the the equilibrium expressed opinions  vary across the nodes in the graph $\Gcal$.
	\begin{definition}%[Controversy~\cite{ChLiDe18}]
		For a graph  $\Gcal= (V,E)$ with expressed opinion vector $\zz$,   the controversy $C(\Gcal,\zz)$ is defined as:
		\begin{equation}\label{eq:dfn_contr}
		C(\Gcal,\zz) =  \sum\limits_{i \in V}\zz _i^2 =\zz^{\top} \zz.
		\end{equation}
	\end{definition}
	The controversy $C(\Gcal,\zz)$ is also introduced as the polarization index proposed in~\cite{MaTeTs17}, but it is normalized by the node number $n$.
	\begin{definition}%[Resistance~\cite{ChLiDe18}]
		For a graph  $\Gcal= (V,E)$ with the internal opinion vector $\vs$ and expressed opinion vector $\zz$, The resistance $\Ical(\Gcal,\zz)$ is the inner product of $\vs$ and $\zz$:
		\begin{equation}\label{eq:dfn_DisCon}
		\Ical(\Gcal,\zz) = \sum_{i\in V}\vs_i \zz_i = \vs^\top \zz.
		\end{equation}
	\end{definition}
	
	The resistance is seemingly close to the sum of controversy and disagreement (also called external conflict) in~\cite{MuMuTs18}, where the authors use the mean-centered opinion vector. Disagreement is defined as $\textstyle \sum_{{(i,j)} \in E} (z_i-z_j)^2$,  characterizing the extent to which acquaintances disagree with each other in their expressed opinions. In~\cite{MuMuTs18}, an algorithm for optimizing the network topology was also developed to reduce resistance for a given internal opinion vector $\vs$. 
	%In addition, for zero mean $\vs$ (and hence zero mean $\zz$), the resistance is also equivalent to the polarization-disagreement index~\cite{MuMuTs18}.
	
	Convenient matrix-vector expressions for the above quantities were provided in~\cite{ChLiDe18,MuMuTs18,XuBaZh21}. For simplicity, we  use $C(\Gcal)$ and $C$ interchangeably to denote $C(\Gcal,\zz)$ if incurring no confusion,  and use $\Ical(\Gcal)$ and $\Ical$ to denote $\Ical(\Gcal,\zz)$.
	%The proof is elementary and omitted for brevity.
	\begin{proposition}\cite{ChLiDe18,MuMuTs18,XuBaZh21}.
		$C(\Gcal)$ and $\Ical(\Gcal)$ can be conveniently expressed
		in terms of quadratic forms as
		\begin{align*}
		C(\Gcal) =\zz^{\top} \zz=\vs^{\top}(\II+\LL)^{-2}\vs,
		\Ical(\Gcal)=\vs^{\top}(\II+\LL)^{-1}\vs.
		\end{align*}
	\end{proposition}
	These two measures 	$C(\Gcal)$ and $\Ical(\Gcal)$ can be written in a unified form 	as $f=\vs^{\top} \MM_{f} \vs$, where $f$ is $C$ or $\Ical$, $\MM_{C}=(\II+\LL)^{-2}$ and $\MM_{\Ical}=(\II+\LL)^{-1}$. In the sequel, we use $\MM$ to represent $\MM_{f}$ when there is no confusion.
	%\begin{remark}
	%In this paper, we will generally assume that the internal opinions $\vs$ are mean-centered, that is, $\bar{\vs}=\vs$. Note that in such case, $\zz$ will also be mean-centered. In fact, as will be shown later, whether the opinions $\vs$ are mean-centered or not, the edges which we select to augment the network are the same since the variation of our objective remains unchanged.
	%\end{remark}
	%described by a graph $\Gcal(V, E)$, where nodes correspond to individuals and edges model social interactions among individuals. With the definitions in Section~\ref{Sec1}, we discuss how the polarization-disagreement index can be optimized by adding edges.
	
	\section{Problem Formulation}
	
	In this section, we first introduce the optimization problem we are concerned with. Then, we study the characterizations for the objective function of the problem. In particular, we show that the object function is monotone and supermodular.
	
	%Then we formally introduce the problem for optimizing resistance and controversy  in a social network.
	
	\subsection{Problem Definition}
	
	In Section~\ref{Sec1}, we quantify various types of conflict for a given internal opinion vector $\vs$. Here we study how to minimize conflict by optimally selecting a set of individuals to change their internal opinions. We focus on minimizing resistance and controversy and use $f(T)$ to denote them, when the opinion $\vs_i$ of every chosen node $i$ in the target set $T$ is modified. Mathematically, in Problem~\ref{prob:pdmi} we formulate our optimization problem for minimizing resistance and controversy in a unifying framework. 
	
	%Note $f_C=\vs^{\top} (\II+\LL)^{-1} \vs$ and $f_{\Ical}=\vs^{\top} (\II+\LL)^{-2} \vs$.
	
	\begin{tcolorbox}
		\begin{problem}\label{prob:pdmi}
			Given a graph $\Gcal= (V,E)$, a vector of internal opinions $\vs$, and an integer $k$, identify a set $T\subset V$ of $k$ nodes and change their internal opinions to $0$, in order to minimize the resistance or controversy $f(T)$. That is: 
			\begin{align}
			\argmin_{T\subset V, \sizeof{T}=k}f(T).
			\end{align}
		\end{problem}
	\end{tcolorbox}
	%Note that  if the initial opinion of a node close to $0$, the benefit of changing the initial opinion of such nodes is relatively low. Thus, we will not inspect all nodes, but focus on a small number of candidate nodes. Specifically, we set a threshold $\vs_i > \alpha$ to determine the candidate node set.
	% Such an optimization problem has been the subject of many  recent papers~\cite{GaKlTa20,MuMuTs18}. Yet we focus on  controlling only a few individuals’ internal opinions rather than manipulating  all individuals’ innate opinions due to the following reason: in practice, manipulating  individuals’ opinions can typically be made only in small amounts, either because of budget constraints, or because of practical considerations. Similar idea was previously used in~\cite{ChLiDe18}.
	%\section{\na\"{\i}veGreedy: A na\"{\i}ve Greedy Algorithm}
	%In this section, we  study the characterization of our problem in more detail. In particular, we reveal two desirable properties of object  function $f(\cdot)$, monotonicity and supermodularity, which give rise to a na\"{\i}ve greedy algorithm for solving Problem~\ref{prob:pdmi} with a guaranteed $(1-\frac{1}{e})$ approximation ratio.
	
	\subsection{Problem Characterization}
	
	The main challenge of Problem~\ref{prob:pdmi} is searching for the promising node subset with the maximum decrease of the objective, which is  inherently a combinatorial problem. %Solving it in a na\"{\i}ve brute-force manner is computationally infeasible. 
	Another obstacle of Problem~\ref{prob:pdmi} is assessing the impact of any given subset of nodes upon the objective. This involves the operations of matrix inversion and multiplication of matrix and vector, which need cubic time and square time, respectively. Specifically, there are all the $\tbinom{n}{k}$ possible sets $T$ for the na\"{\i}ve brute-force method solving Problem~\ref{prob:pdmi}, which results in an exponential complexity $O\big(\tbinom{n}{k}\cdot n^2\big)$ in total. In view of the combinatorial nature of Problem~\ref{prob:pdmi}, it is computationally challenging even for moderately sized networks by  the na\"{\i}ve brute-force method.
	
	To tackle the exponential complexity, we resort to greedy heuristics. Below we show that the objective function of  Problem~\ref{prob:pdmi} has two desirable properties, monotonicity and supermodularity.
	\begin{proposition}[Monotonicity]\label{th:mono}
		$f(T)$  is a monotonically non-increasing function of the node set $T$. In other words,  for any two subsets $T\subseteq W\subseteq V$, one  has $f(T) \geq f(W).$
	\end{proposition}
	\begin{proof}
		To prove the monotonicity, we define a function $g(x,y)$,  $ 1\geq x\geq 0$ and $1\geq y\geq 0$, as
		\begin{align*}
		g(x,y)\overset{def}{=}(\II-x\EE_{ii}-y\EE_{jj})^\top \vs^\top \MM \vs (\II-x\EE_{ii}-y\EE_{jj}),
		\end{align*}where $\EE_{ii}=\ee_i \ee_i^\top$.
		Differentiating the function $g(x,y)$ with respect to $x$, we obtain
		\begin{align*}
		g_x(x,y)=-2\vs^\top\EE_{ii}\MM(\II-x\EE_{ii}-y\EE_{jj})\vs.
		\end{align*}
		It is easy to verify that the entries
		in matrix $\EE_{ii}\MM(\II-x\EE_{ii}-y\EE_{jj})$  and vector $\vs$ are nonnegative, leading to $g_x(x,y)\leq 0$.
	\end{proof}
	%    \Delta(i)=&(\vs-\ee_i\ee_i^\top \vs)^\top\OM(\vs-\ee_i\ee_i^\top \vs) =2(\vs^\top \OM \ee_i)(\ee_i^\top \vs)-(\ee_i^\top \OM \ee_i)(\ee_i^\top \vs)^2\\=&\vs_i(2\zz_i-\omega_{i,i}\vs_i)$\zz_i=\sum^n_{j=1}  \omega_{i,j}\vs_j \geq\omega_{i,i}\vs_i$
	
	Next, we show that function $f(T)$ is supermodular.
	\begin{proposition}[Supermodularity]\label{th:sub}
		$f(T)$  is a supermodular function of the node set $T$. In other words,  for any two subsets $T\subseteq W\subseteq V$ and any node $i\in V\backslash W$, one  has 
		\begin{align}
		\label{m1}
		f(T)- f(T+i)\geq f(W)-f(W+i).  
		\end{align}
	\end{proposition}
	\begin{proof}
		We first prove that for any pair of nodes $i$ and $j$  in $V$, W		\begin{align}
		\label{m2}
		f(\emptyset)- f(\{i\})\geq f( \{j\})-f(\{i, j\}).
		\end{align}
		To this end, we  prove 
		\begin{align}
		\label{m3}
		g(0,0)- g(x,0)\geq g(0,y)- g(x,y),
		\end{align}
		since~(\ref{m3}) is reduced to~(\ref{m2}) in the case of $x=y=1$. In order to prove~(\ref{m3}), it suffices to prove $g_{x,y}(x,y)\geq 0$.
		By successively differentiating function $g(x,y)$ with respect to $x$ and $y$, we obtain $g_{x,y}(x,y)=2\vs^\top\EE_{ii}\MM \EE_{jj}\vs=2\MM_{ij}\vs_i\vs_j$. Since the entries of matrix $\MM$  and vector $\vs$ are nonnegative, one has $g_{x,y}(x,y)\geq 0$.
		Iteratively applying~(\ref{m2}) yields~(\ref{m1}).
		%\begin{align*}
		%   \Ical(T)- \Ical(T\cup {i})=&\vs_i(2\zz_i-\omega_{i,i}\vs_i)
		%\end{align*}
		%$\zz_i(T) \geq \zz_i(W)$
	\end{proof}
	
	\section{Algorithms}
	
	To tackle the challenge of Problem~\ref{prob:pdmi}, we  propose a na\"{\i}ve greedy algorithm with a $(1-1/e)$ approximation guarantee to solve the problem. Then, we integrate several approximation strategies to develop an improved greedy algorithm, which has a $(1-1/e-\epsilon)$ approximation ratio for an error parameter $\epsilon >0$. This fast algorithm is able to significantly accelerate the na\"{\i}ve one, while has little effect on the solution quality in practice.
	
	\subsection{Na\"{\i}ve Greedy Approach}
	
	%In this way, Problem~\ref{prob:pdmi} can be transformed into picking up the edge with highest $f(\cdot)$ value in the candidate set, the algorithms of which will be discussed in the next subsections. 
	
	%Hence, our foremost concerns now are: (a) the efficient assessment of the potential impact of candidate edges addition to the network upon the objective; and (b) the theoretical approximation guarantee of our simple greedy algorithm. These concerns are addressed in what follows.
	%This implies that a
	%small number of edges are being the sources of most—or, at least, a large number of—“good” candidate edges. Intuitively, the node-pairs having the largest difference of expressed opinion should be those edge sources; which  should have the largest impact upon the Polarization-Disagreement Index in the network. 
	
	Our na\"{\i}ve greedy algorithm, denoted as \naiveGreedy, exploits the diminishing returns property of supermodular functions. It first assesses the marginal gain of every candidate node $i$, that is, the decrease of controversy and resistance when the initial opinion of $i$ is changed to 0, and then iteratively adds the most promising node to the solution set $T$ until the budget is reached. 
	
	We present the outline of such a greedy strategy in Algorithm~\ref{alg:Exact}. Initially, the solution node set $T$ is empty.  Then $k$ nodes from set $V\setminus T$ are iteratively selected and added to set $T$.  At each iteration of the na\"{\i}ve greedy algorithm, the node $i$ in candidate set $V\setminus T$ is chosen, which has the largest marginal gain $\Delta(i)=f(T)-f\kh{T+i}$. The algorithm stops until there are $k$ nodes in $T$. 
	
	To evaluate $\Delta(i)$ in Algorithm~\ref{alg:Exact}, it requires computing the inverse of matrix $\II+\LL$ at the beginning, which needs $O(n^3)$ time. Then it performs $k$ rounds, with each round computing $\Delta(i)$ for all candidate nodes $i\in V\setminus T$ in $O(n^2)$ time (Line 3). Thus, the total running time of Algorithm~\ref{alg:Exact} is $O(n^3+kn^2)$. Based on the well-established result in~\cite{NeWoFi78}, Algorithm~\ref{alg:Exact} yields a $(1 - 1/e)$-approximation to the optimal solution to Problem~\ref{prob:pdmi}.
	\begin{theorem}
		The node set $T$ returned by the na\"{\i}ve greedy Algorithm~\ref{alg:Exact} satisfies
		%	\begin{align*}
		$f(\emptyset)-f(T)   \geq (1 - \frac{1}{e})  (f(\emptyset)-f(T_{\rm opt}))$,
		%\end{align*}
		where $T_{\rm opt}$ is the optimal solution to Problem~\ref{prob:pdmi}.
		%\normalsize
	\end{theorem}
	%,  i.e. ,
	%	\begin{align*}
	%	OPT \,  \,  \defeq \underset{T\subset V, \,  \sizeof{T}=k}{\mathrm{arg\, max}} \quad f(T).
	%	\end{align*}
	%A na\"{\i}ve greedy algorithm takes time $O(k|E_C|  n^3)$,  which is computationally intractable even for  small-size networks. Actually, as shown in the proof of Lemma~\ref{lem:dpedc}, with $(\II+\LL)^{-1}$  already computed, we can view the addition of a single edge $e$ as a rank-$1$ update to the original matrix $(\II+\LL)^{-1}$, which can be calculated with a run-time of $O(n^2)$ by using Sherman-Morrison formula~\cite{Me73}, instead of inverting the matrix again in $O(n^3)$ in each loop.
	%Below we show that the computation time can be reduced to $O(n^3)$ by using Sherman-Morrison formula.
	
	%The above analysis leads to a simple greedy algorithm $\na\"{\i}veGreedy(\Gcal, k, \vs)$, which is outlined in Algorithm~\ref{alg:Exact}. To begin with, this algorithm requires $O(n^3)$ time to compute the inverse of $\II+\LL$, and then it performs in $k$ rounds,  with each round mainly including two steps: computing $\Delta(i)$ (Line 4)  in $O(|E_C|n^2)$ time,  and updating  $(\II+\LL)^{-1}$ (Line 8) in $O(n^2)$ time.  Thus,  the total running time of Algorithm~\ref{alg:Exact} is $O(n^3+k|E_C|n^2)$,  which is much  faster than the na\"{\i}ve   algorithm.

	%%%%%%%%%%%%%%%%%%%%%%%%%%%%%%%%%%%%%%
	\begin{algorithm}
		\caption{$\naiveGreedy(\Gcal,  k, \vs)$}
		\label{alg:Exact}
		\Input{
			A connected graph $\Gcal$; an integer $k \leq |V|$; an initial opinion vector $\vs$
		}
		\Output{
			A subset of $T \subset V$ and $|T| = k$
		}
		Initialize solution $T = \emptyset$ \;
		\For{$i = 1$ to $k$}{
			Compute $\Delta(i)\gets f(T)-f\kh{T+i}$
			for each $i \in V \setminus T$ \;
			Select $i$ s.t. $i \gets \mathrm{arg\, max}_{i \in V \setminus T} \Delta(i)$ \;
			Update solution $T \gets T+i$ \;
			Update the opinion vector $\vs \gets \vs-\ee_i\ee_i^\top \vs$ 
		}
		\Return $T$
	\end{algorithm}

	\subsection{Nearly-Linear Time Algorithm}\label{S6}
	%Compared with the na\"{\i}ve algorithm, computation time of Algorithm~\ref{alg:Exact} is significantly reduced. 
	The na\"{\i}ve greedy approach in  Algorithm~\ref{alg:Exact} is  computationally unacceptable for large networks with millions of nodes,  since it requires computing the inverse of  matrix $\II+\LL$. In this subsection, we address this challenge by presenting a fast approximation algorithm $\FastGreedy$ that avoids inverting the matrix $\II+\LL$, and is computationally efficient to solve Problem~\ref{prob:pdmi} in time $\Otil (mk\eps^{-2})$ for any parameter $\eps>0$.  %We first provide its mathematical framework that leads to an $\eps$-approximation  estimator, and then present the outline of $\FastGreedy$.
	
	\subsection{Evaluation of Marginal Gains}
	
	The main computational workload of Algorithm~\ref{alg:Exact} is calculating the marginal gain or impact score $\Delta(i)$ of each node $i\in V$. To solve this computational bottleneck, we first provide a new expression for the marginal gain $\Delta(i)$ of a single node $i$ when its initial opinion is modified.
	\begin{lemma}
		For any node $i\in V\setminus T$,
		\begin{align}
		\Delta(i)=&f(T)-f\kh{T+i}= \vs^\top \MM \vs -(\vs-\ee_i\ee_i^\top\vs)^\top \MM (\vs-\ee_i\ee_i^\top\vs) \nonumber\\ 
		=&\vs_i(2\vs^\top\MM\ee_i-\vs_i\ee_i^\top\MM\ee_i).\label{eq:dpdec}
		\end{align}
	\end{lemma}
	By~\eqref{eq:dpdec}, in order to evaluate $\Delta(i)$, we can alternatively estimate two terms $\vs^\top\MM\ee_i$ and $\MM_{ii}=\ee_i^\top \MM \ee_i$. Below we provide efficient approximations to these two quantities. We first approximate the diagonal entry $\MM_{ii}$ for each $i$. Note that for controversy and resistance, $\MM$ corresponds to $\MM_{C}=\OM^{2}$ and $\MM_{\Ical}=\OM$, respectively. Then, $\MM_{ii}=\ee_i^\top \MM \ee_i$ can be written, respectively, as 
	\begin{align}\label{key}
	\ee_i^\top \OM \ee_i
	&= \ee_i^\top \OM \kh{\II+\BB^\top  \BB} \OM \ee_i  =
	\norm{\OM \ee_i}^2 +
	\norm{\BB \OM \ee_i}^2\, \quad {\rm and} \nonumber\\  
	\ee_i^\top \OM^{2} \ee_i&= \norm{\OM \ee_i}^2.\nonumber
	\end{align}
	
	In this way, we have reduced the estimation of $\MM_{ii}=\ee_i^\top \MM \ee_i$ in~\eqref{eq:dpdec} to the calculation of the $\ell_2$ norms $\norm{\BB \OM \ee_i}^2$ and $\norm{ \OM \ee_i}^2$ of vectors in $\mathbb{R}^{m}$ and $\mathbb{R}^{n}$. Nevertheless, the complexity for exactly computing these two $\ell_2$ norms is still high. Here, we convert to approximate evaluation of these two $\ell_2$ norms using Johnson-Lindenstrauss (JL) lemma~\cite{JoLi84, Ac01}. The JL lemma states that if one projects a set of vectors $\vv_1, \vv_2, \cdots, \vv_n \in \mathbb{R}^d$ (like the columns of matrix $\OM$) onto the $p$-dimensional subspace spanned by the columns of a random matrix $\RR_{p \times d}$	with entries being $1/\sqrt{p}$ or $-1/\sqrt{p}$, where $p \geq 24\log n/\eps^2$ for any given $0<\eps <1$, then the distances between the vectors in the set are nearly preserved with tolerance $1 \pm \eps$. That is,
	\[
	(1-\eps)\norm{\vv_i - \vv_j}^2\leq	\norm{\RR \vv_i - \RR \vv_j}^2 \leq (1+\eps)\norm{\vv_i - \vv_j}^2,
	\]
	holds with probability at least $1 - 1/n$.
	%In order to avoid inverting matrix $\II+\LL$,  we will utilize the fast SDDM linear system solvers \cite{SpTe14} to approximate the above two terms, whose performance could be characterized by the following lemma.	
	
	Let $\QQ_{p\times m}$ and $\PP_{p\times n}$ be two random matrices with entries being $\pm1/\sqrt{p}$, where $p=O(\log n)$. Then we can simply project the column vectors in matrices $\BB \OM$ and $\OM$ onto vectors in low-dimensional vectors in column spaces of $\QQ\BB \OM$ and $\PP\OM$. By JL lemma, we can provide bounds for $\ell_2$ norms of $\BB \OM \ee_i$ and $\OM \ee_i$. However, this still does not help to reduce the computation time, since direct computation of the above $\ell_2$ norms involves inversion of matrix $\II+\LL$. 
	
	In order to avoid computing the inverse of matrix $\II+\LL$, we leverage a nearly linear-time estimator in~\cite{SpTe14} to solve some linear systems. Considering  $\OM=(\II+\LL)^{-1}$, the product $\QQ\BB\OM$ is in fact a solution of the linear system $\XX (\II+\LL)=\QQ\BB$. For a matrix $\XX$, we write $\XX_i$ to denote the $i$-th row of $\XX$. Then, we can solve a linear system of $p=O(\log n)$ equations to obtain $\QQ\BB\OM$, instead of solving a system of $n$ equations required for computing $(\II+\LL)^{-1}$. The solution of each linear system can be obtained efficiently by the fast estimator for a symmetric, diagonally-dominant (SDD) linear system designed for an SDD M-matrix~\cite{SpTe14}, which exploits the approach of preconditioned conjugate gradients to give the unique solution $\XX_i=(\II+\LL)^{-1}(\QQ\BB)_i$. Let $\Solver(\SSS,  \bb,  \eps)$ be SDD linear system estimator, which takes an SDDM matrix $\SSS_{n\times n}$ with $m$ nonzero entries, a vector $\bb \in \mathbb{R}^n$, and an error parameter $\delta > 0$, and return a vector $\yy = \Solver(\SSS,  \bb,  \delta)$ satisfying 
	\begin{align}\label{solver}
	\norm{\yy - \SSS^{-1} \bb}_{\SSS} \leq \delta \norm{\SSS^{-1} \bb}_{\SSS}
	\end{align}
	with  probability at least $1-1/n$,  where $\norm{\yy}_{\SSS} \defeq \sqrt{\yy^\top \SSS \yy}$. The estimator runs in expected time $\Otil (m)$, where $\Otil (\cdot)$ notation suppresses the ${\rm poly} (\log n)$ factors.

	% %, CoKyMiPaJaPeRaXu14
	% 	\begin{lemma}[Fast SDDM Solver \cite{SpTe14}]\label{lem:solver}
	% 		There is a nearly linear time solver $\yy = \Solver(\SSS,  \bb,  \eps)$ which takes an SDDM matrix $\SSS_{n\times n}$ with $m$ nonzero entries,  a vector $\bb \in \mathbb{R}^n$,  and an error parameter $\delta > 0$,  and returns a vector $\xx \in \mathbb{R}^n$ satisfying $\norm{\yy - \SSS^{-1} \bb}_{\SSS} \leq \delta \norm{\SSS^{-1} \bb}_{\SSS}$ with high probability,  where $\norm{\yy}_{\SSS} \defeq \sqrt{\yy^\top \SSS \yy}$.  The solver runs in expected time $\Otil (m)$ .
	% 	\end{lemma}
	
	In order to facilitate the description of the following text, we introduce the notation of $\eps$-approximation for $\eps>0$. For two non-negative scalars $a$ and $b $, we say $a$ is an $\eps$-approximation of $b$ if $(1-\eps) a \leq b \leq (1+\eps) a$, denoted by $a \approx_{\eps} b$. According to~\eqref{solver}, the above estimator can be used to establish an $\eps$-approximation to $\MM_{ii}$ in~\eqref{eq:dpdec}, by properly choose the parameter $\delta$. Let $\XXtil_i$ be the $i$-th row of matrix $\XXtil$ with $\XXtil_i=\Solver(\II+\LL, (\QQ\BB)_i,  \delta_1)$, where 
	%\begin{small}\begin{align}\label{del1}
	$\delta_1 \leq 
	\frac{\eps\sqrt{1-\eps/12}(n-1)}{32 n^2(n+1)\sqrt{(1+\eps/12)(n+1)n}}.$
	%	\end{align}
	%\end{small} 
	Then the term $\norm{\QQ \BB \OM \ee_i}^2$ can be efficiently approximated as stated in the following lemma. 
	\begin{lemma}\label{lem:appro1}
		Given an undirected graph $\Gcal=(V,E)$ with Laplacian matrix $\LL$, a  parameter $\epsilon \in (0, \frac{1}{2})$.
		Then, $\norm{\BB\OM  \ee_i}^2 \approx_{\eps/12} \norm{\XXtil \ee_i}^2$ holds
		for any $i\in V$ with probability almost $1-1/n$.
	\end{lemma}
	\begin{proof}
		On one hand, applying triangle inequality, we obtain
		\begin{small}
			\begin{align*}
			&|\norm{\tilde{\XX}\ee_i}-\norm{\QQ\BB\OM\ee_i}| \le \norm{ \tilde{\XX}\ee_i-\QQ\BB\OM\ee_i }\le \norm{\tilde{\XX}-\QQ\BB\OM}_F\\
			%=\sqrt{\sum_{i=1}^p \|\tilde{\XX}_i-\QQ\BB\OM_i\|^2}
			&  \le\sqrt{\sum_{i=1}^p \|\tilde{\XX}_i-\QQ\BB\OM_i\|_{\LL+\II}^2}	\le\sqrt{\sum_{i=1}^p\delta_1^2\|\QQ\BB\OM_i\|^2_{\LL+\II}} \le \delta_1 \sqrt{n+1}\|\QQ\BB\OM\|_F\\
			& \le \delta_1 \sqrt{n+1} \sqrt{\sum_{i=1}^n(1+\eps/32)\|\QQ\BB\ee_i\|^2}
			%\le \delta_1 \sqrt{n+1}\sqrt{1+\epsilon/32} \sqrt{\sum_{i=1}^n \ee_i^\top \left(\LL+\II\right)^{-1}\ee_i}\\
			\le\delta_1 \sqrt{n} \sqrt{n+1}\sqrt{1+\eps/32}.
			\end{align*}
		\end{small}
		On the other hand, we provide a lower bound of $\norm{\QQ\BB\OM\ee_i}^2$ as
		\begin{align*}
		&\norm{\QQ\BB\OM\ee_i}^2 \geq (1-\frac{\eps}{12})\norm{\BB\OM\ee_i}^2
		=(1-\frac{\eps}{12})\ee_i^\top\OM\LL\OM \ee_i \\
		\geq& (1-\frac{\eps}{12})\ee_i^\top(\II+\LL)^{-1}\LL(\II+\LL)^{-1} \ee_i \\
		=&(1-\frac{\eps}{12})(\ee_i-\frac{1}{n}\one)^\top\OM\LL\OM (\ee_i-\frac{1}{n}\one)
		\geq(1-\frac{\eps}{12})\frac{(n-1)^2}{n^4(n+1)^2}.
		\end{align*}
		% the last inequality is obtained according to the following reason. Note that $\ee_i-\frac{1}{n}\one$ is orthogonal to all-one vector $\one$, an eigenvector
		% of $\LL$ associated with the only eigenvalue $0$. Therefore, $(\ee_i-\frac{1}{n}\one)^\top\LL(\ee_i-\frac{1}{n}\one)\geq \lambda_{\min}
		% \norm{\ee_i-\frac{1}{n}\one}^2$ holds, where the  nonzero
		% minimum eigenvalue $\lambda_{\min}$ of $\LL$ satisfy $\lambda_{\min} \geq 1/(n)^2$~\cite{SpSr11,LiSc18}. In addtion, $\LL$ and $(\II+\LL)^{-1}$ share identical eigenspaces, so $\LL(\II+\LL)^{-1}=(\II+\LL)^{-1}\LL$.
		
		Combining the above-obtained results, it follows that
		\begin{small}
			\begin{align*}
			\frac {\abs{\norm{\tilde{\XX} \ee_i}-\norm{\QQ\BB\OM \ee_i}}}{\norm{\QQ\BB\OM \ee_i}}
			\leq \frac{\delta_1(n+1)n^2\sqrt{n}\sqrt{1+n}\sqrt{1+\frac{\eps}{12}}}{\sqrt{1-\frac{\eps}{12}}(n-1)}\leq \frac{\eps}{32},
			\end{align*}
		\end{small}
		based on which we further obtain
		\begin{small}
			\begin{align}
			&\abs{\norm{\tilde{\XX} \ee_i}^2-\norm{\QQ\BB\OM\ee_i}^2}\nonumber 
			=\abs{\norm{\tilde{\XX} \ee_i}-\norm{\QQ\BB\OM \ee_i}} \times \abs{\norm{\tilde{\XX} \ee_i}+\norm{\QQ\BB\OM \ee_i}}\nonumber\\
			\leq & \frac{\eps}{32}(2+\frac{\eps}{32})\norm{\QQ\BB\OM \ee_i}^2\leq \frac{\eps}{12}\norm{\QQ\BB\OM \ee_i}^2, \nonumber
			\end{align}
		\end{small}
		finishing the proof.
	\end{proof}
	% %, and suppose
	% 			\begin{align*}
	% 		(1-\frac{\eps}{12})\norm{\BB\OM \ee_i}^2 \leq \norm{\QQ\BB\OM \ee_i}^2\leq (1+\frac{\eps}{12})\norm{\BB\OM \ee_i}^2 
	% 		   \end{align*}
	% 	holds for any node $j \in V$
	% , and suppose
	%  	\[(1-\frac{\eps}{12})\norm{\OM \ee_i}^2 \leq \norm{\PP\OM \ee_i}\leq (1+\frac{\eps}{12})\norm{ \OM \ee_i}^2 \]
	%  	for any node $j \in V$
	
	Similarly, we can deal with the case for the term $\norm{\PP \OM \ee_i}^2$.  Let $\YYtil_i$ be the $i$-th row  of matrix $\YYtil$ with $\YYtil_i=\Solver(\II+\LL,
	\PP_i, \delta_2)$, where
	%\begin{small}\begin{align}\label{del2}
	$	\delta_2 \leq \frac{\eps\sqrt{1-\eps/12}}{32(n+1)\sqrt{(1+\eps/12)(n+1)n}}.$
	% 	\end{align}\end{small} 
	%Then, the following lemma  provides an efficient approximation to $\norm{\PP \OM \ee_i}^2$.
	\begin{lemma}\label{lem:appro2}
		Given an undirected graph $\Gcal=(V,E)$ with Laplacian matrix $\LL$, a  parameter $\epsilon \in (0, \frac{1}{2})$.
		Then, $\norm{ \OM \ee_i}^2  \approx_{\eps/12} \norm{\YYtil \ee_i}$ holds for any $i \in V$ with probability almost $1-1/n$.
	\end{lemma}
	Having Lemmas~\ref{lem:appro1} and~\ref{lem:appro2}, the term $\ee_i^\top \OM \ee_i$ can be efficiently approximated by $\ee_i^\top \OM \ee_i  \approx_{\eps/3} \norm{\XXtil \ee_i}^2+\norm{\YYtil \ee_i}^2$.
	%\subsection{Approximating $\zz_{i}$}
	With respect to the term $\vs^\top\MM\ee_i$, it can also be efficiently approximated using the fast SDDM matrix estimator. 
	\begin{lemma}\label{lem:appro3}
		Given an undirected  graph $\Gcal=(V,E)$ with Laplacian matrix $\LL$, a  parameter $\epsilon \in (0, \frac{1}{2})$, and the internal opinion vector $\vs$, let $\qq = \Solver\kh{\II + \LL, \vs, \delta_3}$ and $\hh = \Solver\kh{\II + \LL, \qq, \delta_3}$, where $\delta_3 \leq \epsilon/(6 |\II+\LL|\sqrt{n(1+n)})$. Then, the following relation holds for any $i\in V$ with probability almost $1-1/n$:
		%\begin{align}\label{ineq1}
		$\vs^\top\OM\ee_i \approx_{\eps/3} \qq_i, \vs^\top\OM^2\ee_i \approx_{\eps/3} \hh_i.$ 
		%\end{align}
	\end{lemma}
	\begin{proof}
		%%First of all,
		%%\begin{align}
		%%&(1-\frac{\eps}{6})(1-\ee_j^{\top}\OM\vs)=(1-\frac{\eps}%%{6})(1-\zz_F(j))\nonumber\\
		%%&\geq(1-\frac{\eps}{6})(1-\zz_F(j)_{max})\textgreater0\nonumber
		%%\end{align}
		According to the estimator,
		$\norm{\qq-\OM\vs}_{\II+\LL}
		\leq \delta_3 \norm{\OM\vs}_{\II+\LL}.$
		Then,  the term $|\ee_i^{\top}\qq-\ee_i^{\top}\OM\vs|$ can be bounded as
		\begin{small}
			\begin{align}
			&|\ee_i^{\top}\qq-\ee_i^{\top}\OM\vs| \leq \norm{\ee_i^{\top}}\norm{\qq-\OM\vs}\leq \norm{\qq-\OM\vs}_{\II+\LL}\nonumber  \leq  \delta_3\norm{\OM\vs}_{\II+\LL}  \nonumber\\
			\leq& \delta_3\sqrt{1+n}\norm{\OM\vs}  =\delta_3\sqrt{1+n}\norm{\zz} 
			\leq  \delta_3\sqrt{(1+n)n}\textstyle \sum^n_{j=1}  \vs_j . \nonumber
			\end{align}
		\end{small}	%=|\ee_i^{\top}(\qq-\OM\vs)| =&\norm{\qq-\OM\vs}
		On the other hand,  one obtains $\ee_i^{\top}\OM\vs=\zz_i=\sum^n_{j=1}  \omega_{ij}\vs_j\geq \frac{1}{|\II+\LL|}\sum^n_{j=1}  \vs_j.$
		Then, it follows that
		$	\frac{|\ee_i^{\top}\qq-\ee_i^{\top}\OM\vs|}{\ee_i^{\top}\OM\vs} \leq \frac{\delta_3\sqrt{(1+n)(n)}\sum^n_{j=1}  \vs_j}{\frac{1}{|\II+\LL|}\sum^n_{j=1}  \vs_j} \leq \frac{\eps}{6},$
		which results in  $\vs^\top\OM\ee_i \approx_{\eps/3} \qq_i$. In a similar way, we can prove $\vs^\top\OM^2\ee_i \approx_{\eps/3} \hh_i$.
	\end{proof}
	%\subsection{Fast Algorithm for Approximating $\Delta(i)$}
	
	Based on Lemmas~\ref{lem:appro1},~\ref{lem:appro2} and~\ref{lem:appro3}, we propose an algorithm $\Approx \Delta$ to approximate $\Delta(i)$ for every node $i$ in set $V$. The outline of  algorithm  $\Approx \Delta$ is shown in Algorithm \ref{alg:comp}, whose  performance is given in Lemma~\ref{lem:comp}.
	\begin{lemma}\label{lem:comp}
		For any $0\leq \eps \leq 1/2$,
		the value $\tilde{\Delta}(i)$ returned by $\Approx \Delta$
		satisfies $\Delta(i) \approx_{\eps} \tilde{\Delta}(i)$
		with probability almost $1-1/n$.
	\end{lemma}

	%%%%%%%%%%%%%%%%%%%%%%%%%%%%%%
	\begin{algorithm}
		\small
		\caption{$\Approx \Delta(\Gcal, \vs, \epsilon)$}
		\label{alg:comp}
		\Input{
			A graph $\Gcal$; an initial opinion vector $\vs$; a real number $0 \leq \epsilon \leq1/2$ \\
		}
		\Output{
			$\{(i, \hat{\Delta}(i)| i \in V\}$
			
		}
		Set $\delta_1$, $\delta_2$, $\delta_3$ according to Lemmas~\ref{lem:appro1},~\ref{lem:appro2} and~\ref{lem:appro3}
		\;
		$p \gets \ceil{24\log n/ (\frac{\eps}{12})^2}$ \;		
		Generate random Gaussian matrices
		$\PP_{p\times n},  \QQ_{p\times m}$\;
		Compute $\QQ\BB$
		by sparse matrix multiplication\;
		% 		Compute approximations
		% 		$\qq$ to
		% 		$  \OM\vs$,
		% 		$\XXtil$ to
		% 		$\vs \defeq \BB \OM$,  and
		% 		$\YYtil$ to
		% 		$\OM\defeq \OM$\;
		\For{$i = 1$ to $p$}{
			$\qq
			\gets \Solver(\II+\LL,  \vs,  \delta_3)$\;
			$\hh
			\gets \Solver(\II+\LL,  \qq,  \delta_3)$\;
			$\XXtil_i
			\gets \Solver(\II+\LL,
			(\QQ\BB)_i,  \delta_1)$\;
			$\YYtil_i
			\gets \Solver(\II+\LL,
			\PP_i,  \delta_2)$
		}
		
		\For{each $i\in V$}{
			compute $\hat{\Delta}_{\Ical}(i) =\vs_i(2\qq_i- (\norm{\XXtil\ee_i}^2 + \norm{\YYtil\ee_i}^2)\vs_i)$\;
			compute $\hat{\Delta}_C(i) =\vs_i(2\hh_i-  \norm{\YYtil\ee_i}^2\vs_i)$
		}
		\Return $\{(i, \hat{\Delta}(i))| i \in V\}$
	\end{algorithm}
	
	\subsection{Nearly-Linear Time Greedy Algorithm}
	
	Exploiting Algorithm~\ref{alg:comp} to approximate $\Delta(i)$, we develop an accelerated greedy algorithm $\FastGreedy(\Gcal, E_C, \vs, k, \epsilon)$ in Algorithm~\ref{alg:Appro} to solve Problem~\ref{prob:pdmi}. As in Algorithm \ref{alg:Exact},  Algorithm~\ref{alg:Appro} performs $k$ rounds (Lines 2-6).  In each round, it takes time $\Otil(m\eps^{-2})$ to call $\Approx \Delta$ to approximate the marginal gain $\hat{\Delta}(i)$ for all candidate nodes $i$, then select the node with the highest impact score and update the target set $T$ and the initial opinion vector $\vs$.  Therefore,  the total time complexity of  Algorithm~\ref{alg:Appro} is $\Otil (mk\eps^{-2})$.

	The following theorem presents that the output $T$ of  Algorithm~\ref{alg:Appro} yields a $(1 - 1/e - \eps)$ approximation solution to Problem~\ref{prob:pdmi}.
	\begin{theorem}
		For any  $0<\epsilon \leq1/2$,  the node set $T$ returned by the greedy Algorithm~\ref{alg:Appro} satisfies
		%\begin{align*}
		$f(\emptyset)-f(T)   \geq (1 - \frac{1}{e} - \eps) (f(\emptyset)-f(T_{\rm opt}))$,
		%	\end{align*}
		with probability almost $1-1/n$, where $T_{\rm opt}$ is the optimal solution to Problem~\ref{prob:pdmi}.
		%\begin{align*}
		%	OPT \,  \,  \defeq \underset{T\subset V, \,  \sizeof{T}=k}{\mathrm{arg\, max}} \quad f(T).
		%	\end{align*}
		\normalsize
	\end{theorem}
	
	%%%%%%%%%%%%%%%%%%%%%%
	\begin{algorithm}[htbp]
		\caption{$\FastGreedy(\Gcal, \vs, k, \epsilon)$}
		\label{alg:Appro}
		\Input{
			A graph $\Gcal$; an initial opinion vector $\vs$; an integer $k \leq |V|$; a real number $0 \leq \epsilon \leq1/2$
		}
		\Output{
			$T$: a subset of $V$ and $|T| = k$
		}
		Initialize solution $T = \emptyset$ \;
		\For{$i = 1$ to $k$}{
			$\{i, \hat{\Delta}(i) | i \in V \setminus T \} \gets \Approx \Delta(\Gcal, \vs, \epsilon)$ \;
			Select $i$ s.t.  $i \gets \mathrm{arg\, max}_{i \in V \setminus T} \hat{\Delta}(i)$ \;
			Update solution $T \gets T+i$ \;
			Update the opinion vector $\vs \gets \vs-\ee_i\ee_i^\top \vs$ 
		}
		\Return $T$
	\end{algorithm}
	%%%%%%%%%%%%%%%%%%%%%%
	
	\section{Experimental  Results}\label{S7}
	
	In this section, we evaluate the performance of our two greedy algorithms $\naiveGreedy$ and $\FastGreedy$. To this end, extensive experiments are designed and executed on various real networks to validate both the effectiveness and efficiency of our algorithms.

	%\begin{table}
	%	\centering
	%	\caption{Statistics of datasets.}\label{netinfo}
	%	\begin{tabular}{lcrllcr}
	%		\toprule
	%		Network & $n$ & $m$ &&Network & $n$ & $m$ \\
	%		\midrule
	%		Karate & 34 & 78 && Books & 105 & 441  \\
	%		ClintonTrump & 2832 & 18551 && 
	%		Polblogs & 1224 & 16718  \\
	%Diseasome && 516 && 1188  \\
	%Yeast && 1,458 && 1,948   \\
	%GridWorm && 3,507 && 6,531  \\
	%		Erdos992 && 5,094 && 7,515  \\
	%		Reality && 6,809 && 7,680  \\
	%		\bottomrule
	%	\end{tabular}
	%\end{table}
	% Please add the following required packages to your document preamble:
	% \usepackage{multirow}
	
	\begin{table*}
		\fontsize{12.0}{12.5}\selectfont
		\setlength\tabcolsep{3pt}
		\centering
		\caption{The running time (seconds, $s$) and the relative error $\Gamma$ ($\times 10^{-2}$) of $\FastGreedy$ and \emph{BOMP}  for minimizing the controversy on various networks with $k=50$ and three  distributions of initial opinions: uniform distribution, power-law distribution, and exponential distribution. The Algorithm $\FastGreedy$ is abbreviated to  $\FastGreedyy$.}\label{table:eff}
		\resizebox{\linewidth}{!}{
			\begin{tabular}{lrrrrrrccccccccccc}
				\toprule
				\multicolumn{1}{l}{\multirow{3}{*}{Network}} & \multicolumn{1}{c}{\multirow{3}{*}{$n$}} & \multicolumn{1}{r}{\multirow{3}{*}{$m$}} & \multicolumn{5}{c}{Uniform distribution} & \multicolumn{5}{c}{Power-law distribution} & \multicolumn{5}{c}{Exponential distribution} \cr
				\cmidrule(lr){4-8}
				\cmidrule(lr){9-13}
				\cmidrule(lr){14-18}
				\multicolumn{1}{c}{} & \multicolumn{1}{c}{} & \multicolumn{1}{c}{} & \multicolumn{2}{c}{Time} & \multicolumn{2}{c}{$\Delta C(\Gcal)$} & \multicolumn{1}{c}{\multirow{2}{*}{$\Gamma$}} & \multicolumn{2}{c}{Time} & \multicolumn{2}{c}{$\Delta C(\Gcal)$} & \multicolumn{1}{c}{\multirow{2}{*}{$\Gamma$ }} & \multicolumn{2}{c}{Time} & \multicolumn{2}{c}{$\Delta C(\Gcal)$} & \multicolumn{1}{c}{\multirow{2}{*}{$\Gamma$}} \cr
				\cmidrule(lr){4-5}
				\cmidrule(lr){6-7}
				\cmidrule(lr){9-10}
				\cmidrule(lr){11-12}
				\cmidrule(lr){14-15}
				\cmidrule(lr){16-17}
				\multicolumn{1}{c}{} & \multicolumn{1}{c}{} & \multicolumn{1}{c}{} & \multicolumn{1}{c}{\emph{BOMP}} & \multicolumn{1}{c}{$\FastGreedyy$} & \multicolumn{1}{c}{\emph{BOMP}} & \multicolumn{1}{c}{$\FastGreedyy$} & \multicolumn{1}{c}{} & \multicolumn{1}{c}{\emph{BOMP}} & \multicolumn{1}{c}{$\FastGreedyy$} & \multicolumn{1}{c}{\emph{BOMP}} & \multicolumn{1}{c}{$\FastGreedyy$} & \multicolumn{1}{c}{} & \multicolumn{1}{c}{\emph{BOMP}} & \multicolumn{1}{c}{$\FastGreedyy$} & \multicolumn{1}{c}{\emph{BOMP}} & \multicolumn{1}{c}{$\FastGreedyy$} & \multicolumn{1}{c}{} \cr
				\midrule
				EmailUniv & 1133 & 5451 & 0.76 & 0.83 & -0.0355 & -0.0352 & 0.79 & 0.74 & 0.85 & \multicolumn{1}{r}{-0.2236} & \multicolumn{1}{r}{-0.2240} & \multicolumn{1}{c}{0.18} & 0.76 & 0.83 & -0.1172 & -0.1171 & 0.06 \cr
				Yeast & 1458 & 1948 & 1.23 & 0.80 & -0.0753 & -0.0744 & 1.15 & 1.27 & 0.83 & \multicolumn{1}{r}{-0.8297} & \multicolumn{1}{r}{-0.8302} & \multicolumn{1}{c}{0.06} & 1.22 & 0.78 & -0.7245 & -0.7233 & 0.18 \cr
				Hamster & 2426 & 16630 & 3.56 & 1.78 & -0.0538 & -0.0537 & 0.20 & 3.68 & 1.83 & \multicolumn{1}{r}{-0.5495} & \multicolumn{1}{r}{-0.5507} & \multicolumn{1}{c}{0.21} & 3.65 & 1.77 & -0.2823 & -0.2822 & 0.05 \cr
				GrQc & 4158 & 13422 & 9.60 & 2.91 & -0.0273 & -0.0274 & 0.35 & 9.75 & 2.85 & \multicolumn{1}{r}{-0.0875} & \multicolumn{1}{r}{-0.0875} & \multicolumn{1}{c}{0.06} & 9.93 & 2.84 & -0.0719 & -0.0718 & 0.19 \cr
				Erdos992 & 5094 & 7515 & 14.68 & 2.99 & -0.0202 & -0.0198 & 1.70 & 15.23 & 3.08 & \multicolumn{1}{r}{-0.9250} & \multicolumn{1}{r}{-0.9246} & \multicolumn{1}{c}{0.05} & 14.56 & 3.19 & -0.8569 & -0.8561 & 0.09 \cr
				PagesGovernment & 7057 & 89455 & 26.61 & 5.72 & -0.0089 & -0.0088 & 1.47 & 26.76 & 5.59 & \multicolumn{1}{r}{-0.5885} & \multicolumn{1}{r}{-0.5891} & \multicolumn{1}{c}{0.11} & 25.83 & 5.78 & -0.5127 & -0.5119 & 0.17 \cr
				AstroPh & 17903 & 196972 & 250.01 & 29.51 & -0.0056 & -0.0056 & 0.42 & 246.42 & 29.96 & \multicolumn{1}{r}{-0.0929} & \multicolumn{1}{r}{-0.0932} & \multicolumn{1}{c}{0.34} & 255.61 & 30.04 & -0.0819 & -0.0816 & 0.39 \cr
				CondMat & 21363 & 91286 & 372.50 & 29.00 & -0.0062 & -0.0063 & 1.65 & 372.17 & 28.79 & \multicolumn{1}{r}{-0.4684} & \multicolumn{1}{r}{-0.4687} & \multicolumn{1}{c}{0.07} & 378.62 & 29.23 & -0.4388 & -0.4381 & 0.17 \cr
				Gplus & 23628 & 39194 & 489.20 & 24.64 & -0.0027 & -0.0026 & 2.74 & 501.21 & 24.91 & \multicolumn{1}{r}{-1.8446} & \multicolumn{1}{r}{-1.8422} & \multicolumn{1}{c}{0.13} & 500.37 & 25.04 & -1.5905 & -1.5886 & 0.12 \cr
				GemsecRO$*$ & 41773 & 125826 & --- & 51.51 & --- & --- & --- & --- & 50.86 & --- & --- & --- & --- & 52.44 & --- & --- & --- \cr
				WikiTalk$*$ & 92117 & 360767 & --- & 93.85 & --- & --- & --- & --- & 96.27 & --- & --- & --- & --- & 92.98 & --- & --- & --- \cr
				Gowalla$*$ & 196591 & 950327 & --- & 238.76 & --- & --- & --- & --- & 244.54 & --- & --- & --- & --- & 234.22 & --- & --- & --- \cr
				GooglePlus$*$ & 211187 & 1143411 & --- & 263.92 & --- & --- & --- & --- & 267.27 & --- & --- & --- & --- & 269.04 & --- & --- & --- \cr
				MathSciNet$*$ & 332689 & 820644 & --- & 334.14 & --- & --- & --- & --- & 337.63 & --- & --- & --- & --- & 330.96 & --- & --- & --- \cr
				Flickr$*$ & 513969 & 3190452 & --- & 665.08 & --- & --- & --- & --- & 684.82 & --- & --- & --- & --- & 684.09 & --- & --- & --- \cr
				IMDB$*$ & 896305 & 3782454 & --- & 1248.2 & --- & --- & --- & --- & 1248.0 & --- & --- & --- & --- & 1296.1 & --- & --- & --- \cr
				YoutubeSnap$*$ & 1134890 & 2987624 & --- & 1139.6 & --- & --- & --- & --- & 1111.9 & --- & --- & --- & --- & 1121.0 & --- & --- & --- \cr
				Flixster$*$ & 2523386 & 7918801 & --- & 2207.5 & --- & --- & --- & --- & 2221.4 & --- & --- & --- & --- & 2208.3 & --- & --- & ---\cr
				\bottomrule
			\end{tabular}
		}
	\end{table*}
	
	\begin{table*}
		\fontsize{12.0}{12.5}\selectfont
		\setlength\tabcolsep{3pt}
		\centering
		\caption{The running time (seconds, s) and the relative error $\Gamma$  ($\times 10^{-2}$) of $\FastGreedy$ and $\naiveGreedy$ for minimizing the resistance on various networks with $k=50$ and three  distributions of initial opinions: uniform distribution, power-law distribution, and exponential distribution. The Algorithm $\FastGreedy$ is abbreviated to  $\FastGreedyy$. }\label{table:eff2}
		\resizebox{\linewidth}{!}{
			\begin{tabular}{lrrrrrrccccccccccc}
				\toprule
				\multicolumn{1}{l}{\multirow{3}{*}{Network}} & \multicolumn{1}{c}{\multirow{3}{*}{$n$}} & \multicolumn{1}{r}{\multirow{3}{*}{$m$}} & \multicolumn{5}{c}{Uniform distribution} & \multicolumn{5}{c}{Power-law distribution} & \multicolumn{5}{c}{Exponential distribution} \\
				\cmidrule(lr){4-8}
				\cmidrule(lr){9-13}
				\cmidrule(lr){14-18}
				\multicolumn{1}{c}{} & \multicolumn{1}{c}{} & \multicolumn{1}{c}{} & \multicolumn{2}{c}{Time} & \multicolumn{2}{c}{$\Delta \Ical(\Gcal)$} & \multicolumn{1}{c}{\multirow{2}{*}{$\Gamma$ }} & \multicolumn{2}{c}{Time} & \multicolumn{2}{c}{$\Delta \Ical(\Gcal)$} & \multicolumn{1}{c}{\multirow{2}{*}{$\Gamma$ }} & \multicolumn{2}{c}{Time} & \multicolumn{2}{c}{$\Delta \Ical(\Gcal)$} & \multicolumn{1}{c}{\multirow{2}{*}{$\Gamma$}} \\
				\cmidrule(lr){4-5}
				\cmidrule(lr){6-7}
				\cmidrule(lr){9-10}
				\cmidrule(lr){11-12}
				\cmidrule(lr){14-15}
				\cmidrule(lr){16-17}
				\multicolumn{1}{c}{} & \multicolumn{1}{c}{} & \multicolumn{1}{c}{} & \multicolumn{1}{c}{$\naiveGreedy$} & \multicolumn{1}{c}{$\FastGreedyy$} & \multicolumn{1}{c}{$\naiveGreedy$} & \multicolumn{1}{c}{$\FastGreedyy$} & \multicolumn{1}{c}{} & \multicolumn{1}{c}{$\naiveGreedy$} & \multicolumn{1}{c}{$\FastGreedyy$} & \multicolumn{1}{c}{$\naiveGreedy$} & \multicolumn{1}{c}{$\FastGreedyy$} & \multicolumn{1}{c}{} & \multicolumn{1}{c}{$\naiveGreedy$} & \multicolumn{1}{c}{$\FastGreedyy$} & \multicolumn{1}{c}{$\naiveGreedy$} & \multicolumn{1}{c}{$\FastGreedyy$} & \multicolumn{1}{c}{} \\
				\midrule
				EmailUniv & 1133 & 5451 & 0.20 & 1.05 & -92.28 & -91.12 & 1.26 & 0.19 & 0.95 & -47.03 & -46.73 & 0.65 & 0.18 & 1.00 & -44.50 & -44.44 & 0.13 \\
				Yeast & 1458 & 1948 & 0.26 & 1.26 & -149.67 & -147.02 & 1.77 & 0.28 & 1.19 & -58.06 & -57.95 & 0.19 & 0.25 & 1.22 & -51.21 & -51.10 & 0.21 \\
				Hamster & 2426 & 16630 & 0.79 & 2.56 & -143.47 & -142.36 & 0.77 & 1.19 & 2.65 & -33.24 & -32.99 & 0.74 & 1.27 & 2.55 & -33.45 & -33.40 & 0.13 \\
				GrQc & 4158 & 13422 & 2.91 & 4.76 & -153.56 & -151.11 & 1.60 & 2.83 & 4.30 & -33.99 & -33.93 & 0.17 & 2.84 & 4.31 & -30.24 & -30.14 & 0.35 \\
				Erdos992 & 5094 & 7515 & 4.27 & 5.50 & -147.28 & -144.19 & 2.10 & 3.98 & 4.50 & -43.15 & -43.09 & 0.14 & 4.14 & 4.68 & -43.19 & -42.98 & 0.49 \\
				PagesGovernment & 7057 & 89455 & 8.98 & 9.33 & -114.38 & -112.30 & 1.82 & 8.39 & 8.48 & -42.85 & -42.77 & 0.20 & 8.45 & 8.64 & -42.14 & -42.05 & 0.21 \\
				AstroPh & 17903 & 196972 & 95.18 & 24.63 & -143.67 & -140.07 & 2.51 & 93.16 & 22.62 & -41.93 & -41.65 & 0.65 & 95.15 & 22.26 & -42.99 & -42.89 & 0.25 \\
				CondMat & 21363 & 91286 & 153.73 & 26.44 & -178.65 & -176.29 & 1.32 & 152.73 & 23.47 & -44.83 & -44.74 & 0.21 & 155.06 & 23.39 & -46.19 & -46.10 & 0.20 \\
				Gplus & 23628 & 39194 & 205.43 & 25.09 & -114.31 & -110.01 & 3.77 & 203.73 & 21.93 & -56.26 & -55.82 & 0.77 & 203.83 & 22.66 & -50.86 & -50.51 & 0.68 \\
				GemsecRO$*$ & 41773 & 125826 & --- & 52.04 & --- & --- & --- & --- & 46.04 & --- & --- & --- & --- & 44.49 & --- & --- & --- \\
				WikiTalk$*$ & 92117 & 360767 & --- & 108.66 & --- & --- & --- & --- & 102.13 & --- & --- & --- & --- & 107.05 & --- & --- & --- \\
				Gowalla$*$ & 196591 & 950327 & --- & 277.21 & --- & --- & --- & --- & 259.57 & --- & --- & --- & --- & 255.14 & --- & --- & --- \\
				GooglePlus$*$ & 211187 & 1143411 & --- & 293.94 & --- & --- & --- & --- & 272.63 & --- & --- & --- & --- & 270.55 & --- & --- & --- \\
				MathSciNet$*$ & 332689 & 820644 & --- & 431.20 & --- & --- & --- & --- & 408.55 & --- & --- & --- & --- & 383.20 & --- & --- & --- \\
				Flickr$*$ & 513969 & 3190452 & --- & 787.22 & --- & --- & --- & --- & 740.83 & --- & --- & --- & --- & 756.99 & --- & --- & --- \\
				IMDB$*$ & 896305 & 3782454 & --- & 1862.6 & --- & --- & --- & --- & 1641.7 & --- & --- & --- & --- & 1596.6 & --- & --- & --- \\
				YoutubeSnap$*$ & 1134890 & 2987624 & --- & 1621.8 & --- & --- & --- & --- & 1456.0 & --- & --- & --- & --- & 1452.0 & --- & --- & --- \\
				Flixster$*$ & 2523386 & 7918801 & --- & 3446.3 & --- & --- & --- & --- & 3318.0 & --- & --- & --- & --- & 3299.1 & --- & --- & --- \\
				\bottomrule
			\end{tabular}
		}
	\end{table*}
	
	\subsection{Experiment Setup}
	
	\textbf{Datasets.} The studied realistic networks are publicly available in the KONECT~\cite{kunegis2013konect} and SNAP~\cite{LeSo16}. For each network, we implement our experiments on its largest component. The first three columns of Table~\ref{table:eff} show the relevant statistics of the networks.\\
	%The characteristics of the largest components for all networks are summarized in
	\textbf{Machine and reproducibility.} All algorithms in our experiments are executed in Julia. In our algorithm $\FastGreedy$, we use the linear estimator $\Solver$~\cite{kyng2016approximate},  the Julia implementation  of  which is available on the website\footnote{https://github. com/danspielman/Laplacians.jl}.  All experiments were conducted on a machine equipped with 32G RAM and 4.2 GHz Intel i7-7700 CPU.\\
	%Our code is publicly available~\footnote{ https://anonymous.4open.science/r/dp-3152}.  \\
	\textbf{Node selection strategies.} The sets of $k$ nodes are determined using the following five strategies. (1) \emph{Random}: selecting $k$ nodes from $V$ at random. (2) \emph{PageRank}: selecting top-$k$ nodes with the highest PageRank scores~\cite{Br01}. (3) $\FastGreedy$: selecting $k$ nodes by algorithm $\FastGreedy$. (4) $\naiveGreedy$: selecting $k$ nodes using algorithm $\naiveGreedy$. (5) \emph{BOMP}: selecting $k$ nodes using the strategy in~\cite{MaTeTs17}.  \\
	%each element $\vs_i$ of the initial opinions vector $\vs$ is uniformly distributed in the range of $[0,1]$.
	\textbf{Opinions and evaluation metrics.} In our experiments, the internal	opinions are generated according to three different distributions: uniform distribution, exponential distribution, and power-law distribution.  The performance of the five strategies for node selection is evaluated by their impacts on the drop of controversy and resistance denoted, respectively, by $\Delta C(\Gcal)$  and $\Delta \Ical(\Gcal)$, with a larger decrease corresponding to an more effective method for node selection. For the approximation algorithm $\FastGreedy$, we set the parameter $\eps= 0.5$. Note that one can adjust $\eps$ to achieve a balance between effectiveness and efficiency, with a smaller value of $\eps $ corresponding to better effectiveness but relatively poor efficiency. In all of our experiments $\eps= 0.5$ is enough to guarantee both good effectiveness and efficiency. 
	
	\subsection{Comparison of  Effectiveness}
	
	We first evaluate the effectiveness of our algorithms $\naiveGreedy$ and $\FastGreedy$ for optimizing the resistance, by comparing them with \emph{PageRank}  and the random scheme \emph{Random}. For this purpose, we execute experiments on four realistic networks: Karate with 34 nodes and  78 edges, Books with 105 nodes and  441 edges,  ClintonTrump with 2832 nodes and  18551 edges, and Polblogs with 1224 nodes and 16718 edges. In Figure~\ref{ComOpt1}, we show how the resistance is decreased, when different strategies are used to select nodes. As can be seen from Figure~\ref{ComOpt1}, $\naiveGreedy$ always returns the best result as expected, and $\FastGreedy$ is very close to that of $\naiveGreedy$. Both of proposed algorithms consistently outperform the schemes of \emph{PageRank}  and  \emph{Random}.

	%In fact, we only showed the results uniformly random distribution of the initial opinions, due to the space limitations. Actually, 
	%The performance of different methods for varying $k$  are  displayed in  Figures~\ref{ComOpt1}.
	\begin{figure}
		\centering
		\includegraphics[width=\linewidth]{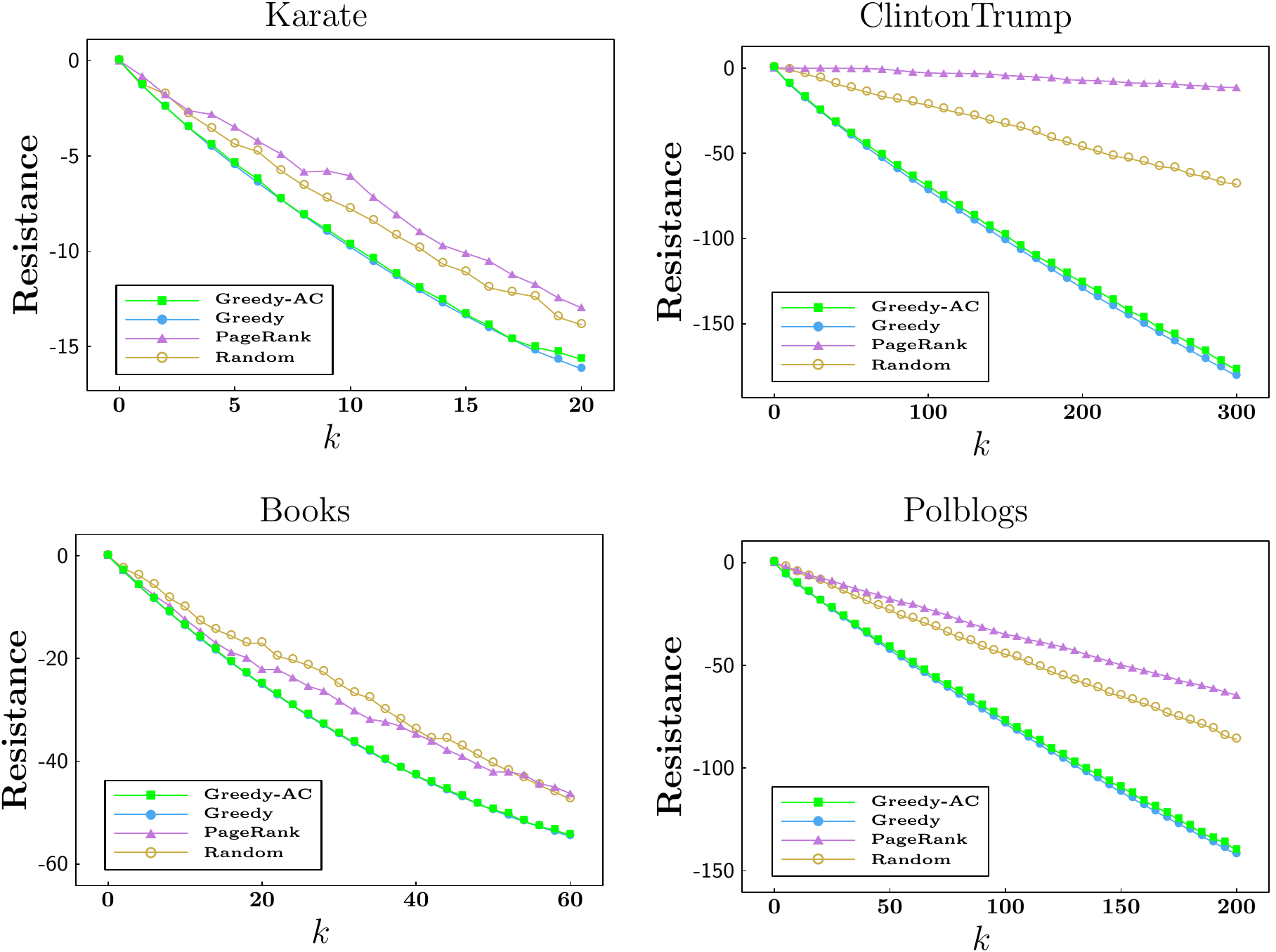}
		\caption{Resistance after performing  four methods for node selection on datasets:  Karate, Books,  ClintonTrump and Polblogs for varying $k$. The initial opinions of nodes obey a uniform distribution.}\label{ComOpt1}  
	\end{figure}
	
	\begin{figure}
		\centering
		\includegraphics[width=\linewidth]{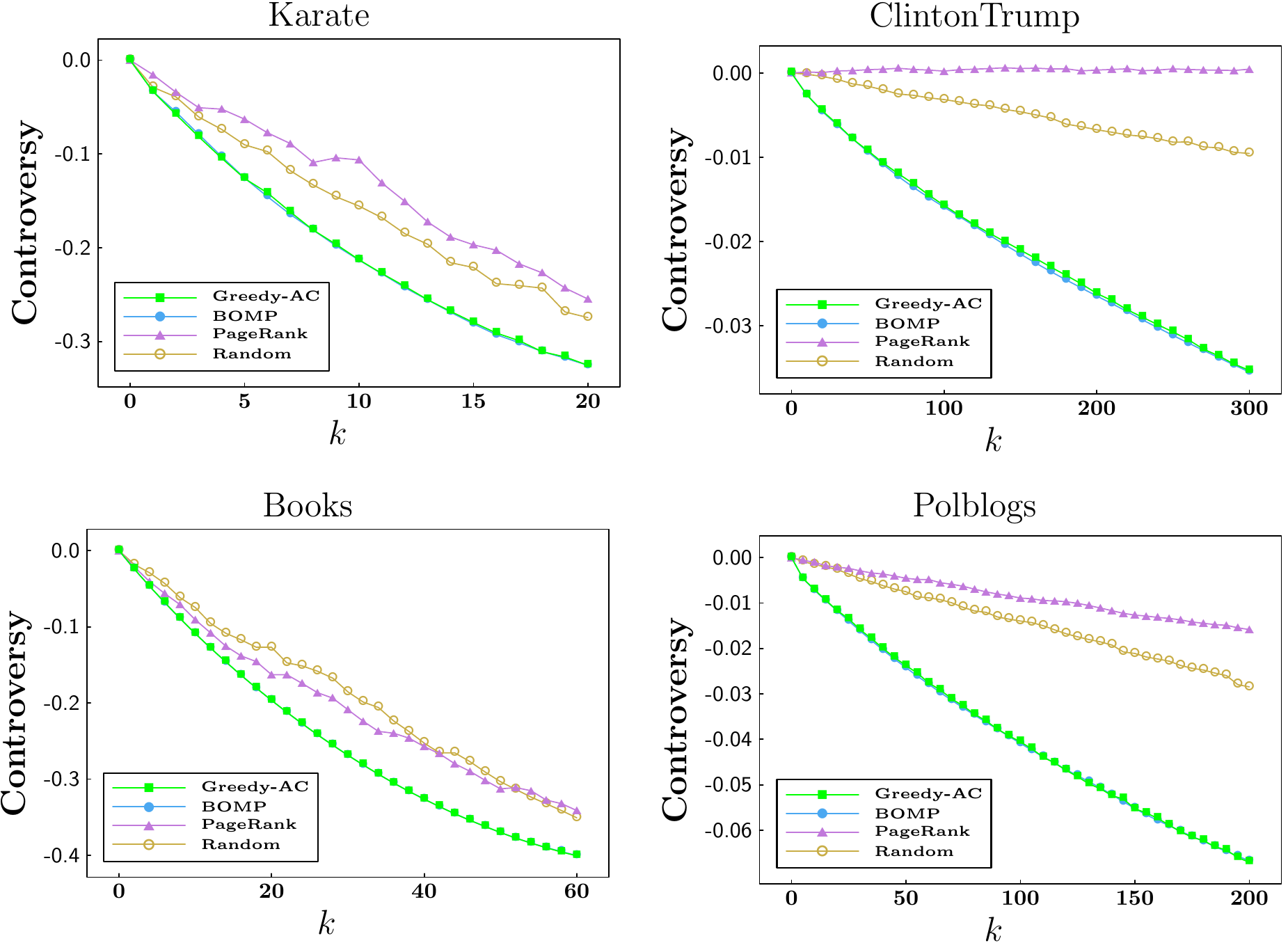}
		\caption{Controversy after performing four methods for node selection on datasets:  Karate,  Books,  ClintonTrump and  Polblogs for varying $k$. The initial opinions of nodes obey a uniform distribution. }\label{ComOpt2}
	\end{figure}

	We continue to demonstrate the effectiveness of our algorithm  $\FastGreedy$ for optimizing the controversy,  by comparing it with three baseline schemes, \emph{PageRank}, \emph{Random}, and \emph{BOMP} in~\cite{MaTeTs17}. Figure~\ref{ComOpt2} illustrates the results for the four methods of node selection on the same four networks as in Figure~\ref{ComOpt1}. From Figure~\ref{ComOpt2}, one can observe that the decrease of the controversy yielded by $\FastGreedy$ and \emph{BOMP} are almost the same, and are very close to the optimal solutions according to the experimental results reported in~\cite{MaTeTs17}. Moreover, both $\FastGreedy$ and \emph{BOMP} are significantly better than  \emph{PageRank}  and  \emph{Random}.
	
	We note that Figures~\ref{ComOpt1} and~\ref{ComOpt2} only report the results for the case that the initial opinions of nodes follow a uniform distribution. For the cases that initial opinions obey an exponential distribution or a power-law distribution, the results are similar to those in Figures~\ref{ComOpt1} and~\ref{ComOpt2}. We omit these results due to the space limit.
	
	%$\na\"{\i}veGreedy$ achieves the best performance as expected, and the proposed $\FastGreedy$ (a) is very close to the \emph{BOMP} method, and (b) consistently outperforms all two alternative methods.
	
	\subsection{Comparison of Running Time }%of Greedy Algorithms
	Although both $\FastGreedy$ and \emph{BOMP} achieve remarkable effectiveness for optimizing the controversy, we now show that $\FastGreedy$ runs much faster than \emph{BOMP}. To this end, we compare the running time of $\FastGreedy$ with that of \emph{BOMP} on $18$ real-world networks. For each network, we select $k = 50$ nodes to minimize the controversy by using \emph{BOMP} and $\FastGreedy$, respectively. In Table~\ref{table:eff}, we provide the results of running time and the drop of the controversy $\Delta C(\Gcal)$ for the two strategies. Table~\ref{table:eff} shows that $\FastGreedy$ is significantly faster than \emph{BOMP} for those networks with more than 1400 nodes. The improvement of efficiency of $\FastGreedy$ over \emph{BOMP} becomes more significant when the graphs grow in size, and the speed-up of $\FastGreedy$ is up to $20\times$. It is worth noting that \emph{BOMP} is not applicable to the last nine networks marked with "$*$" due to the limitations of time and memory. In comparison, $\FastGreedy$ is scalable to large networks with more than $10^6$ nodes.
	
	In spite the fact that in comparison with \emph{BOMP} our $\FastGreedy$ algorithm achieves significant improvement in the terms of efficiency, we will show the results returned by $\FastGreedy$ are close to those for \emph{BOMP}, besides the four aforementioned networks. To show this, we measure the relative error $\Gamma= |\tilde{\beta}-\beta|/\beta$ of the result for $\FastGreedy$ on every network in Table~\ref{table:eff}, where $\beta$ and $\tilde{\beta}$ are the decrease of  controversy $\Delta C(\Gcal)$ corresponding to $\FastGreedy$ and \emph{BOMP}, respectively. From the relative errors reported in Table~\ref{table:eff}, we observe that these relative errors are negligible for all tested networks, with the largest value equal to $2.74\%$.  Thus, the results returned by $\FastGreedy$ are very close to those associated with \emph{BOMP}, implying that $\FastGreedy$ is both effective and efficient, independent on the  distributions of initial opinions.
	
	We also present an extensive comparison of the performance of our two algorithms $\naiveGreedy$ and $\FastGreedy$ for optimizing the resistance, in terms of the efficiency and effectiveness. In Table~\ref{table:eff2}, we report their running time  and relative errors on various real-world networks. Table~\ref{table:eff2} indicates that $\FastGreedy$ returns similar results as $\naiveGreedy$, but runs much faster than  $\naiveGreedy$. Thus, $\FastGreedy$ always achieves ideal performance irrespective of the distributions of initial opinions, in the contexts of both efficiency and effectiveness.

	\section{Conclusion}\label{S8}
	
	In this paper, we addressed the problem of minimizing risk of conflict, including controversy and resistance, by strategically changing the initial opinions of a small number of individuals. We unified the two optimization problems into one framework, and showed that the objective function is monotone and supermodular. We then presented two greedy algorithms to solve the optimization problem. The former returns a $(1-1/e)$ approximation to the optimum solutions in time $O(n^3)$, while the latter provides a $(1-1/e -\eps)$ approximation in time $\Otil (mk\eps^{-2})$ for a positive parameter $\eps$. On the theoretic side, we provided  detailed analysis of the approximation guarantee for the latter algorithm. On the experimental side, we performed extensive experiments on real-life networks, demonstrating that both of our algorithms lead to almost optimal solutions, and consistently outperform several alternative baseline heuristics. Particularly, our second algorithm yields a good approximation solution quickly on networks with more than two million nodes within $40$ minutes, demonstrating excellent scalablity to large-scale networks.
	
	\section*{Acknowledgements}
	The work was supported by the Shanghai Municipal Science and Technology
	Major Project  (No.2018SHZDZX01), the
	National Natural Science Foundation of China ( No.61872093),  and  ZJLab.
	
	%This problem belongs to the class of discrete optimization that has found vast applications in various domains.
	
	%It deserves to mention that although we only considered unweighted networks, our approximate algorithms can be easily generalized or modified to weighted networks. On the other hand, ~\cite{ChLiDe18} implied that different measures of phenomena related to opinion dynamics act like communicating vessels: reducing one implies that another one must be increased. Hence, our approach can be easily adapted to optimize other quantities.
	%On the other hand, most phenomena related to opinion dynamics are quantified and studied based  on undirected graphs, in future work, it would be interesting to address these problems for directed graphs.
	\bibliographystyle{ACM-Reference-Format}
	\balance
	\bibliography{kdd2022}
	
	%In this section, we provide detailed proofs of Lemmas~\ref{lem:sub},~\ref{lem:appro1},~\ref{lem:appro3}, and Theorem~\ref{lem:comp}. Note, for two matrices $\AAA$ and $\BB$,  we write $\AAA \preceq \BB$ to denote that $\BB - \AAA$ is positive semidefinite. We use $(\II+\LL)_T^{-1}$ to denote the forest matrix associated with graph $\Gcal+T$.
\end{document}